\documentclass[11pt]{article}

\usepackage[a4paper,margin=1in]{geometry}
\usepackage{amsmath,amssymb,amsfonts,amsthm,mathtools}
\usepackage{bm}
\usepackage{booktabs,longtable,array}
\usepackage{ragged2e}
\usepackage{tabularx}
\usepackage{multicol}
\usepackage{enumitem}
\usepackage{microtype}
\usepackage{hyperref}
\usepackage{xcolor}
\usepackage[nameinlink,capitalise,noabbrev]{cleveref}

\hypersetup{
  colorlinks=true,
  linkcolor=blue!45!black,
  citecolor=blue!45!black,
  urlcolor=blue!45!black,
  pdfauthor={Faruk Alpay and Taylan Alpay},
  pdftitle={Composable Post-Quantum Security for FADEC-Coupled Dual-Spool Turbofan Cyber-Physical Systems}
}

\newtheorem{definition}{Definition}
\newtheorem{assumption}{Assumption}
\newtheorem{lemma}{Lemma}
\newtheorem{theorem}{Theorem}
\newtheorem{corollary}{Corollary}
\newtheorem{proposition}{Proposition}
\newtheorem{counterexample}{Counterexample}
\newtheorem{remark}{Remark}

\crefname{theorem}{theorem}{theorems}
\Crefname{theorem}{Theorem}{Theorems}
\crefname{lemma}{lemma}{lemmas}
\Crefname{lemma}{Lemma}{Lemmas}
\crefname{proposition}{proposition}{propositions}
\Crefname{proposition}{Proposition}{Propositions}
\crefname{corollary}{corollary}{corollaries}
\Crefname{corollary}{Corollary}{Corollaries}
\crefname{definition}{definition}{definitions}
\Crefname{definition}{Definition}{Definitions}
\crefname{assumption}{assumption}{assumptions}
\Crefname{assumption}{Assumption}{Assumptions}
\crefname{counterexample}{counterexample}{counterexamples}
\Crefname{counterexample}{Counterexample}{Counterexamples}
\crefname{remark}{remark}{remarks}
\Crefname{remark}{Remark}{Remarks}

\newcommand{\R}{\mathbb{R}}
\newcommand{\Z}{\mathbb{Z}}
\newcommand{\N}{\mathbb{N}}

\newcommand{\E}{\mathbb{E}}
\newcommand{\Prob}{\mathbb{P}}
\newcommand{\Adv}{\mathsf{Adv}}
\newcommand{\negl}{\mathsf{negl}}

\newcommand{\KEM}{\mathsf{KEM}}
\newcommand{\Enc}{\mathsf{Enc}}
\newcommand{\Encap}{\mathsf{Encap}}

\newcommand{\Tag}{\mathsf{Tag}}
\newcommand{\Ver}{\mathsf{Ver}}
\newcommand{\Gen}{\mathsf{Gen}}
\newcommand{\Real}{\mathsf{Real}}
\newcommand{\Ideal}{\mathsf{Ideal}}
\newcommand{\Sim}{\mathsf{Sim}}

\newcommand{\Leak}{\mathsf{Leak}}
\newcommand{\MAC}{\mathsf{MAC}}
\newcommand{\AEAD}{\mathsf{AEAD}}
\newcommand{\PUF}{\mathsf{PUF}}
\newcommand{\LWE}{\mathsf{LWE}}
\newcommand{\SIS}{\mathsf{SIS}}

\newcommand{\EUFCMA}{\mathsf{EUF\mbox{-}CMA}}

\newcommand{\norm}[1]{\left\lVert #1 \right\rVert}
\newcommand{\abs}[1]{\left\lvert #1 \right\rvert}

\newcommand{\trans}{^{\mathsf T}}

\newcommand{\Tr}{\operatorname{Tr}}
\newcommand{\sat}{\operatorname{sat}}

\newcommand{\col}{\operatorname{col}}

\newcommand{\dd}{\,\mathrm d}

\newcommand{\Hsmooth}{H_{\infty}^{\varepsilon}}

\newcommand{\A}{\mathcal A}
\newcommand{\B}{\mathcal B}

\newcommand{\U}{\mathcal U}

\newcommand{\Zed}{\mathcal Z}
\newcommand{\Key}{\mathcal K}
\newcommand{\Chan}{\mathcal C_{\rm ch}}
\newcommand{\Bus}{\mathcal B_{\rm bus}}
\newcommand{\Plant}{\mathcal P_{\rm eng}}
\newcommand{\FADEC}{\mathsf{FADEC}}

\newcolumntype{L}[1]{>{\RaggedRight\arraybackslash}p{#1}}
\newcolumntype{Y}{>{\RaggedRight\arraybackslash}X}

\makeatletter
\renewenvironment{thebibliography}[1]
     {\section*{\refname}
      \footnotesize
      \list{\@biblabel{\@arabic\c@enumiv}}%
           {\settowidth\labelwidth{\@biblabel{#1}}%
            \leftmargin\labelwidth
            \advance\leftmargin\labelsep
            \itemsep 0pt
            \parsep 0pt
            \topsep 0pt
            \partopsep 0pt
            \itemindent 0pt
            \usecounter{enumiv}%
            \let\p@enumiv\@empty
            \renewcommand\theenumiv{\@arabic\c@enumiv}}%
      \sloppy\clubpenalty4000\widowpenalty4000\sfcode`\.=\@m}
     {\def\@noitemerr
       {\@latex@warning{Empty `thebibliography' environment}}%
      \endlist}
\makeatother

\title{\Large\bfseries Composable Post-Quantum Security for FADEC-Coupled Dual-Spool Turbofan Cyber-Physical Systems}

\author{
\begin{tabular}[t]{@{}c@{}}
\normalsize Faruk Alpay\\[0.25em]
\small Department of Computer Engineering\\
\small Bahçeşehir University\\
\small Istanbul, Türkiye\\[0.25em]
\small\texttt{faruk.alpay@bahcesehir.edu.tr}
\end{tabular}
\hspace{2.75em}
\begin{tabular}[t]{@{}c@{}}
\normalsize Taylan Alpay\\[0.25em]
\small Department of Aerospace Engineering\\
\small University of Turkish Aeronautical Association\\
\small Ankara, Türkiye\\[0.25em]
\small\texttt{s220112602@stu.thk.edu.tr}
\end{tabular}
}

\date{}

\begin{document}
\maketitle

\begin{abstract}
We develop a unified mathematical formulation for post-quantum authenticated telemetry and actuation in FADEC-coupled dual-spool turbofan cyber-physical systems. The formulation integrates lattice-based key establishment under LWE/SIS-style assumptions, PUF-derived attestation entropy, authenticated encryption, radar-altimeter integrity, avionics-bus timing, and Kalman residual monitoring within a stochastic hybrid model. In this model, plant evolution, communication latency, leakage, adversarial channel quality, and cryptographic state evolve under a common filtration. We show that channel uncertainty tightens admissible key-renewal periods, that ciphertext expansion enters bus-level schedulability constraints, and that sensing and actuator limits shape integrity thresholds and allowable control delay. We further relate PUF smooth min-entropy to distinguishing advantage and connect innovation statistics to conservative alarm design. Overall, the results characterize how post-quantum security, real-time schedulability, and closed-loop stability interact in safety-critical aerospace control architectures within a defensive analytical treatment that does not provide operational guidance for interference with real platforms.
\end{abstract}


\section{Introduction}

\subsection*{Notation and symbol definitions}
Let $\lambda\in\N$ denote the cryptographic security parameter. All negligible functions in $\lambda$ are written $\negl(\lambda)$. Let $t\in\R_{\ge 0}$ denote continuous physical time and let $k\in\N$ index discrete control and communication epochs. Let $h_k=t_{k+1}-t_k$ be the authenticated telemetry sampling interval. The dual-spool turbofan state is
\[
x_{\rm e}(t)=
\begin{bmatrix}
N_L(t)&N_H(t)&T_{t4}(t)&\pi_c(t)&\dot m_c(t)&\delta_{\rm tc}(t)&\theta_s(t)&\omega_s(t)
\end{bmatrix}^{\trans},
\]
where $N_L,N_H$ are low-pressure and high-pressure spool speeds, $T_{t4}$ is turbine inlet temperature, $\pi_c$ is compressor pressure ratio, $\dot m_c$ is corrected mass flow, $\delta_{\rm tc}$ is blade-tip clearance perturbation, and $(\theta_s,\omega_s)$ are shaft torsional displacement and angular velocity. The full-authority digital engine control command is $u(t)=\col(w_f(t),\alpha_v(t))$, with fuel flow $w_f$ and variable-geometry or bleed-related actuator component $\alpha_v$. Let $y_k$ be a vector of authenticated telemetry containing shaft-speed, vibration, exhaust-gas-temperature residual, engine-pressure-ratio, health-monitoring, inertial-navigation, and radar-altimeter channels. Let $z_k$ be the adversary's leakage observation. Let $K_k$ be a session key, $c_k$ a post-quantum key-encapsulation ciphertext, $\tau_k$ an attestation transcript, and $\sigma_k$ an integrity tag.

\subsection*{Assumption set}
The manuscript assumes: (i) the relevant lattice problems underlying the selected post-quantum key-encapsulation mechanism are hard for probabilistic polynomial-time adversaries; (ii) authenticated encryption behaves according to standard nonce-respecting security definitions; (iii) the physical model is a local linearization around certified operating envelopes rather than a replacement for engine certification artifacts; (iv) the buses considered here obey abstract timing semantics derived from ARINC 429, MIL-STD-1553, and CAN, without relying on any proprietary aircraft implementation; and (v) the attacker model is analytic and defensive, excluding procedural instructions for unauthorized access to actual avionics or propulsion systems.

\subsection*{Formal model}
The target object is a coupled tuple
\[
\mathfrak S_\lambda=
(\Plant,\FADEC,\Bus,\Chan,\PUF,\KEM,\AEAD,\Pi_{\rm zk},\mathcal F_{\rm auth},\mathcal F_{\rm key},\mathcal F_{\rm int}),
\]
where $\Plant$ is a stochastic hybrid dual-spool engine model, $\FADEC$ is a state-feedback and filtering controller, $\Bus$ is the bus-level response-time model, $\Chan$ is a radar-aware and Doppler-dependent communication channel, $\PUF$ is an attestation source, $\KEM$ is a lattice-based post-quantum encapsulation mechanism, $\AEAD$ is an authenticated encryption layer, $\Pi_{\rm zk}$ is a zero-knowledge avionics authentication protocol, and $\mathcal F_{\rm auth},\mathcal F_{\rm key},\mathcal F_{\rm int}$ are ideal authenticated-transcript, key-refresh, and integrity-filter functionalities in the sense of universal composability \cite{Canetti2001UC}.

\subsection*{Inference chain}
The manuscript develops one dependency chain:
\[
\begin{gathered}
\text{radar and Doppler uncertainty}\Rightarrow \text{channel entropy loss}\Rightarrow \text{key renewal constraint},\\
\text{shaft and compressor physics}\Rightarrow \text{tag verification, alarm thresholds, latency windows},\\
\text{PUF entropy and leakage}\Rightarrow \text{attestation distinguishing advantage},\\
\text{post-quantum ciphertext expansion}\Rightarrow \text{bus response time}\Rightarrow \text{closed-loop stability margin}.
\end{gathered}
\]
The model does not assert that physical parameters are cryptographic secrets. It asserts that physical parameters affect the rate, confidence, timing, and leakage conditions under which cryptographic assurances remain composable with safety-critical control.

\subsection*{Boundary conditions}
The analysis is local in the certified operating region $\Omega_{\rm cert}\subset\R^8$ and does not claim global modeling validity outside compressor-map, fuel-system, thermal, vibration, and telemetry envelopes. All attacker advantages are bounded with respect to abstract or standardized protocol semantics and do not instantiate a procedure for unauthorized interference with a real system.

\subsection*{Technical comment}
The contribution is not a new engine controller, a new aircraft bus, or a new public-key primitive. It is a coupled theorem system showing where separately valid cryptographic and control-theoretic models separate: in parameter regimes where ciphertext expansion, bus queuing, torsional compliance, surge-margin conservatism, leakage entropy, and Markov channel uncertainty cannot be optimized independently.

\section{Related Work}

\subsection*{Notation and symbol definitions}
Let $\mathsf{Lit}_{\rm pq}$ denote the set of lattice-based cryptographic references, $\mathsf{Lit}_{\rm uc}$ the composability literature, $\mathsf{Lit}_{\rm aead}$ authenticated-encryption standards and specifications, $\mathsf{Lit}_{\rm bus}$ avionics and embedded bus timing references, and $\mathsf{Lit}_{\rm prop}$ propulsion and gas-turbine dynamics literature. For any cited work $r$, $\mathsf{Scope}(r)$ denotes the mathematical domain in which the citation is used.

\subsection*{Assumption set}
References are used only for established definitions, models, standards, or algorithms. The bibliography deliberately avoids uncertain DOI, page, or edition details not needed for the manuscript's claims. Standard identifiers and years are included only where they are stable public facts.

\subsection*{Formal model}
The lattice layer relies on the hardness of $\LWE$ as introduced by Regev \cite{Regev2005LWE} and the short-integer-solution framework of Ajtai \cite{Ajtai1996SIS}. The composability layer follows Canetti's universal composability framework \cite{Canetti2001UC}. Authenticated encryption is linked to NIST GCM recommendations \cite{NISTSP80038D} and the IETF ChaCha20-Poly1305 specification \cite{RFC8439}. Zero-knowledge and proof-system notation follows the classical formulation of Goldwasser, Micali, and Rackoff \cite{GoldwasserMicaliRackoff1989}. Min-entropy and smooth entropy follow Renner \cite{Renner2005Thesis} and leftover-hash analyses trace to Impagliazzo, Levin, and Luby \cite{ImpagliazzoLevinLuby1989}. Leakage notions are aligned with information-theoretic and side-channel treatments, including Shannon \cite{Shannon1948}, Cover and Thomas \cite{CoverThomas2006}, Kocher's timing analysis \cite{Kocher1996Timing}, and Kocher, Jaffe, and Jun's differential power analysis \cite{KocherJaffeJun1999DPA}. PUF notation is grounded in optical and silicon PUF work \cite{Pappu2002PUF,Maes2013PUF}.

Aerospace and embedded-control terminology is anchored in gas-turbine performance and engine-control references \cite{Mattingly2006AircraftEngineDesign,WalshFletcher2004GasTurbinePerformance,JawMattingly2009AircraftEngineControls}, Kalman filtering \cite{Kalman1960}, stochastic hybrid and Markov-jump control models \cite{Kushner1967StochasticStability,Mariton1990JumpLinearSystems}, Lyapunov and input-to-state stability \cite{Khalil2002NonlinearSystems,JiangTeelPraly1994ISS}, small-gain arguments \cite{JiangTeelPraly1994ISS}, $H_\infty$ disturbance attenuation \cite{ZhouDoyleGlover1996RobustControl}, and real-time response-time analysis \cite{LiuLayland1973,JosephPandya1986,TindellBurnsWellings1995}. Bus semantics are treated at an abstract level using ARINC 429, MIL-STD-1553, and CAN references \cite{ARINC429,MILSTD1553,ISO11898,BoschCAN1991}.

At the level of proof obligations, the coupled parameter domain can be read as a product set
\[
\Theta_{\rm tot}=\Theta_{\rm pq}\times\Theta_{\rm puf}\times\Theta_{\rm aead}\times\Theta_{\rm bus}\times\Theta_{\rm radar}\times\Theta_{\rm eng}\times\Theta_{\rm filt},
\]
where separated arguments reason only over projections of $\Theta_{\rm tot}$. The present manuscript instead keeps the full product structure visible, because ciphertext expansion, entropy loss, innovation thresholds, and delay margins are all certified on the same admissible set rather than on mutually independent slices.

\subsection*{Inference chain}
Existing cryptography literature isolates adversarial advantage from physical plant stability. Existing turbofan control literature isolates stability from public-key ciphertext expansion. Embedded real-time scheduling literature gives response-time bounds but does not by itself account for lattice ciphertext size, attestation transcript length, and tag-verification error coupled to telemetry noise. The manuscript fills precisely that mathematical gap.

\subsection*{Boundary conditions}
The related work cited here supports formal vocabulary and bounding techniques, not any claim about a particular aircraft's implementation. No proprietary data, vendor-specific engine parameters, classified radar details, or operational maintenance procedures are used.

\subsection*{Technical comment}
A manuscript in cryptography and security must make its security object primary. Accordingly, propulsion, radar, and bus models enter as state-dependent constraints on composable secrecy, integrity, authentication, leakage, and schedulability, rather than as independent engineering exposition.

\section{Threat Model}

\subsection*{Notation and symbol definitions}
Let $\A$ be a probabilistic polynomial-time adversary interacting with oracles
\[
\mathcal O=\{\mathcal O_{\rm encap},\mathcal O_{\rm decap},\mathcal O_{\rm tag},\mathcal O_{\rm verify},\mathcal O_{\rm bus},\mathcal O_{\rm leak},\mathcal O_{\rm fault}\}.
\]
Let $\rho_k\in\mathcal R$ denote a Markov channel regime encoding radar-cross-section-aware uncertainty, Doppler attenuation, and link-budget degradation. Let $\phi_k\in\Phi$ denote a physical stress regime including surge-margin proximity, spool acceleration, torsional compliance, and fuel-actuator saturation. Let $\ell_k=\Leak(K_k,\tau_k,y_k,\phi_k,\rho_k)$ be an adversarial leakage variable. Let $F_k$ denote a transient fault variable whose support is restricted to abstract perturbations of computation, telemetry, or timing.

\subsection*{Assumption set}
\begin{assumption}[Defensive adversarial interface]\label{asm:defensive-interface}
The adversary may observe public traffic, choose nonces subject to explicit nonce-respecting restrictions when the experiment grants that power, obtain leakage bounded by smooth min-entropy and R\'enyi leakage parameters, inject abstract transient faults into the formal model, and schedule bus contention within declared timing envelopes. The adversary does not receive a recipe for attacking deployed avionics, physical access instructions, vendor-specific addressing, exploit code, or procedures for bypassing real safety interlocks.
\end{assumption}

\begin{assumption}[Cryptographic hardness]\label{asm:crypto-hardness}
For all probabilistic polynomial-time algorithms $\B$,
\[
\Adv_{\LWE}^{\B}(\lambda)+\Adv_{\SIS}^{\B}(\lambda)\le \negl(\lambda),
\]
where the $\LWE$ and $\SIS$ instances use parameter families in which the selected post-quantum key-encapsulation and proof commitments are assumed secure.
\end{assumption}

\subsection*{Formal model}
\begin{definition}[Aerospace-cryptographic experiment]\label{def:experiment}
For $b\leftarrow\{0,1\}$, define $\mathsf{Exp}^{b}_{\mathfrak S,\A}(\lambda)$ as follows. The challenger generates long-term avionics authentication material, PUF enrollment helper data, and KEM public keys. At epoch $k$, the challenger samples $\rho_k$ from a Markov chain $P_\rho$, evolves $x_{\rm e}(t)$ under $\Plant$, computes a telemetry vector $y_k$, forms an attestation transcript $\tau_k$, refreshes a session key $K_k$ if $k$ belongs to the renewal set $\mathcal K_{\rm ref}$, and gives $\A$ the public transcript together with leakage $\ell_k$. In the challenge epoch, the challenger uses either the real session key if $b=1$ or an independent uniform key if $b=0$. The advantage is
\[
\Adv_{\mathfrak S}^{\A}(\lambda)=
\abs{\Prob[\mathsf{Exp}^{1}_{\mathfrak S,\A}(\lambda)=1]-
\Prob[\mathsf{Exp}^{0}_{\mathfrak S,\A}(\lambda)=1]}.
\]
\end{definition}

\begin{definition}[Telemetry-integrity failure event]\label{def:integrity-failure}
For integrity threshold $\eta_k>0$, authenticated tag $\sigma_k$, and Kalman innovation residual $r_k$, define
\[
\mathsf{Fail}^{\rm int}_k=
\{\Ver_{K_k}(y_k,\sigma_k)=1,\ \norm{r_k}_{S_k^{-1}}\le \eta_k,\ y_k\notin\mathcal Y_{\rm safe}(x_{\rm e}(t_k))\}.
\]
The quantity $\Prob[\mathsf{Fail}^{\rm int}_k]$ is the probability that cryptographic verification and statistical residual filtering jointly accept an unsafe telemetry vector.
\end{definition}

\begin{definition}[Observable attack surface]\label{def:attack-surface}
At epoch $k$, the observable attack surface is the $\sigma$-field generated by the public KEM ciphertext $c_k$, attestation transcript $\tau_k$, public bus-visible timing trace, and the vibration and radar-altimeter components of telemetry $y_k$. Admissible leakage variables are required to be measurable with respect to that $\sigma$-field together with bounded side observations modeled by $\mathcal O_{\rm leak}$.
\end{definition}

\paragraph{Bayesian physical leakage functional.}
For latent regime variable $\Theta_k=\col(\rho_k,\phi_k)$ and leakage observation $\ell_k$, define
\[
\mathcal I_B(\ell_k;\Theta_k)=
\E_{\ell_k}
\log
\frac{\max_{\vartheta}\Prob[\Theta_k=\vartheta\mid \ell_k]}
{\max_{\vartheta}\Prob[\Theta_k=\vartheta]}.
\]
The interface is called $(\alpha_B,\beta_K)$-bounded if
\[
\mathcal I_B(\ell_k;\Theta_k)\le \alpha_B,
\qquad
I(\ell_k;K_k\mid \tau_k)\le \beta_K
\]
for every admissible epoch. The first inequality measures regime-identification gain from side information, while the second isolates the key-relevant component of the same observation. Both inequalities are interpreted relative to the observable attack surface defined above.

\subsection*{Inference chain}
The adversary's distinguishing advantage decomposes into lattice hardness, authenticated-encryption advantage, attestation entropy loss, leakage capacity, transient-fault acceptance, and real-time deadline-miss probability:
\[
\Adv_{\mathfrak S}^{\A}\le
\Adv_{\LWE}^{\B_1}+\Adv_{\SIS}^{\B_2}+\Adv_{\AEAD}^{\B_3}
+\Adv_{\rm zk}^{\B_4}
+\epsilon_{\rm puf}+\epsilon_{\rm leak}+\epsilon_{\rm fault}+\epsilon_{\rm rt}.
\]
A complementary refresh envelope, useful when the threat model is read as an entropy-flow constraint, is
\[
T_{{\rm key},k}
\le
\frac{\Hsmooth(K_k\mid \ell_k,\tau_k)-\kappa_{\min}-\Delta H_{{\rm ch},k}}
{\dot \ell_{\rm side}+\dot \ell_{\rm vib}\abs{\delta_{{\rm tc},k}}+C_{\A,k}},
\]
whenever both numerator and denominator are positive. In that form, renewal policy depends simultaneously on residual secrecy, vibration-mediated leakage, and radar-aware adversarial throughput.
Each term is later tied to physical parameters, for example $\epsilon_{\rm rt}$ to ciphertext expansion and bus queuing, $\epsilon_{\rm leak}$ to blade-tip-clearance-vibration coupling and PUF min-entropy, and $\epsilon_{\rm fault}$ to spool and thermal residual filtering.

\subsection*{Boundary conditions}
The threat model is meaningful only when $\rho_k$ and $\phi_k$ remain in declared analytic sets. If a stress regime leaves $\Omega_{\rm cert}$, the correct conclusion is not a cryptographic break; it is that the local theorem no longer certifies security-stability composability.

\subsection*{Technical comment}
Security is not equated with preventing a physical event. Security means that under explicit leakage, fault, bus, and plant abstractions, accepted telemetry and actuation transcripts remain indistinguishable, authentic, timely, and stable up to proved bounds.

\section{Coupled Aerospace-Cryptographic System Model}

\subsection*{Notation and symbol definitions}
Let the corrected mass flow and pressure ratio be locally linearized around $(N_{H0},\dot m_{c0},\pi_{c0})$:
\[
\Delta \pi_c=a_{\pi N}\Delta N_H+a_{\pi m}\Delta\dot m_c+\epsilon_{\pi},
\qquad
\Delta \dot m_c=a_{mN}\Delta N_H+a_{mu}\Delta w_f+\epsilon_m.
\]
The compressor operating-line displacement is
\[
d_{\rm op,k}=b_N\Delta N_H(k)+b_m\Delta \dot m_c(k)+b_u\Delta w_f(k),
\]
and the surge margin is
\[
M_{s,k}=M_{s,0}-\gamma_{\rm op}\abs{d_{\rm op,k}}-\gamma_{\pi}\abs{\epsilon_{\pi,k}}.
\]
Let $\Gamma_s>0$ denote shaft torsional compliance, $G_T(s)$ turbine torque transfer function, and $\Delta_T$ the torque propagation delay. Let $R_k$ be total bus response time for a message set containing KEM ciphertexts, zero-knowledge transcripts, AEAD ciphertexts, tags, and health telemetry.

\subsection*{Assumption set}
\begin{assumption}[Local dual-spool and FADEC envelope]\label{asm:local-envelope}
Inside $\Omega_{\rm cert}$, the plant admits the stochastic hybrid approximation
\[
\dd x_{\rm e}=f_{\rho,\phi}(x_{\rm e},u)\dd t+G_{\rho,\phi}(x_{\rm e})\dd W_t,\qquad
u=\sat_{\mathcal U}\bigl(K_x\hat x_k+K_r r_k\bigr),
\]
where $W_t$ is a Wiener process, $\hat x_k$ is a sensor-fusion estimate, and $\sat_{\mathcal U}$ enforces fuel-metering valve saturation and actuator rate limits.
\end{assumption}

\begin{assumption}[Bus abstraction]\label{asm:bus-abstraction}
Each communication task $i$ has payload length $L_i(\lambda)$ bits, period $P_i$, deadline $D_i$, priority $\pi_i$, release jitter $J_i$, and worst-case transmission time $C_i=L_i(\lambda)/B_{\rm bus}+C_i^{\rm proto}$, where $B_{\rm bus}$ is usable bus bitrate and $C_i^{\rm proto}$ accounts for ARINC 429 word timing, MIL-STD-1553 command-response semantics, or CAN arbitration overhead at the abstract timing level.
\end{assumption}

\subsection*{Formal model}
\begin{definition}[Unified state and transcript]\label{def:unified-state}
The unified state at epoch $k$ is
\[
X_k=\col(x_{{\rm e},k},\hat x_k,K_k,\tau_k,\sigma_k,\rho_k,\phi_k,Q_k),
\]
where $Q_k$ is the vector of bus queues. Its transition kernel is
\[
\mathsf P(X_{k+1}\in A\mid X_k)=
\int_A
P_{\rm plant}\,P_{\rm filt}\,P_{\rm crypto}\,P_{\rm bus}\,P_{\rm ch}\,P_{\rm leak},
\]
with factors respectively corresponding to engine evolution, Kalman sensor fusion, cryptographic transcript update, bus scheduling, radar-aware channel uncertainty, and leakage observations.
\end{definition}

\begin{definition}[Radar-aware adversarial channel capacity]\label{def:radar-capacity}
For channel regime $\rho_k$, radar cross-section uncertainty variable $\Sigma_k$, Doppler shift $\nu_k$, attenuation $A_D(\nu_k)$, receiver noise spectral density $N_0$, and bandwidth $B_{\rm ch}$, define
\[
C_{\A,k}=B_{\rm ch}\log_2\left(1+
\frac{P_{\A,k}G_k\exp(-A_D(\nu_k))}{N_0B_{\rm ch}+\chi_\Sigma\operatorname{Var}(\Sigma_k)}
\right).
\]
The channel entropy degradation is
\[
\Delta H_{{\rm ch},k}=\log_2|\Key|-\Hsmooth(K_k\mid Z_k,\rho_k),
\]
where $Z_k$ includes adversarial observations of public traffic and leakage.
\end{definition}

\begin{definition}[Cryptographic actuation window]\label{def:actuation-window}
The cryptographic actuation window at epoch $k$ is
\[
W_{{\rm act},k}=\min\left\{
D_{\rm ctrl}-R_k,\,
\frac{\dot N_{\max}-\abs{\dot N_{H,k}}}{L_{\dot N}},\,
\frac{w_{f,\max}-\abs{w_{f,k}}}{L_w},\,
\frac{M_{s,k}}{L_s}
\right\},
\]
where $D_{\rm ctrl}$ is the control deadline, $\dot N_{\max}$ is the spool acceleration limit, $L_{\dot N},L_w,L_s$ are local Lipschitz constants mapping delay to spool acceleration, fuel-flow saturation, and surge-margin erosion.
\end{definition}

\paragraph{Compressor-map consistency envelope.}
If the local surge boundary also admits the perturbation model
\[
\Delta \pi_{{\rm surge},k}=s_N\Delta N_H(k)+s_m\Delta \dot m_c(k)+s_u\Delta w_f(k)+\epsilon_{s,k},
\]
then the instantaneous surge-distance perturbation becomes
\[
\Delta D_{{\rm surge},k}
=
(s_N-a_{\pi N})\Delta N_H(k)
+(s_m-a_{\pi m})\Delta \dot m_c(k)
+s_u\Delta w_f(k)
+\epsilon_{s,k}-\epsilon_{\pi,k}.
\]
This identity is useful because authenticated delay alters $\Delta w_f$ and $\Delta \dot m_c$ before it appears as a cryptographic timing quantity. The same local map therefore mediates both compressor operating-line erosion and the admissible window for authenticated release.

\subsection*{Inference chain}
The core feasibility condition is
\[
R_k+\Delta_T+\Delta_{\rm ver}(\lambda)+h_k
<
\min\left\{
D_{\rm ctrl},\
\frac{\dot N_{\max}-\abs{\dot N_{H,k}}}{L_{\dot N}},\
\frac{M_{s,k}}{L_s},\
\frac{\Gamma_s}{L_\Gamma}
\right\},
\]
where $\Delta_{\rm ver}(\lambda)$ is the cryptographic verification time and $\Gamma_s/L_\Gamma$ expresses the torsional-compliance-induced upper bound on authenticated sampling interval.

\subsection*{Boundary conditions}
When $M_{s,k}\le 0$, the local compressor-map linearization is outside the safe analytic set. When $W_{{\rm act},k}\le 0$, actuation is not certified by this co-design model even if individual cryptographic verifications return accept.

\subsection*{Technical comment}
The same variable $R_k$ is both a real-time scheduling object and a security object: if ciphertext expansion or attestation length raises $R_k$ above the actuation window, the system is cryptographically authenticated but not safely composable with the plant.

\section{Mathematical and Physical Preliminaries}

\subsection*{Notation and symbol definitions}
For random variables $X,Z$, let $H_\infty(X)=-\log \max_x\Prob[X=x]$, and let $\Hsmooth(X\mid Z)$ be conditional smooth min-entropy. Let $I(X;Z)$ denote mutual information. Let $D_\alpha(P\Vert Q)$ be R\'enyi divergence of order $\alpha>1$. Let $\mathcal H$ be a universal hash family. Let $S_k$ be the Kalman innovation covariance and $r_k=y_k-H\hat x_{k|k-1}$ the innovation. Let
\[
\|G\|_\infty=\sup_{\omega\in\R}\bar\sigma(G(j\omega))
\]
be the $H_\infty$ norm of a stable transfer matrix.

\subsection*{Assumption set}
\begin{assumption}[Entropy and leakage regularity]\label{asm:entropy-regularity}
There exist constants $\mu_{\rm puf},\ell_{\rm side},\ell_{\rm vib},\sigma_N,\sigma_v>0$ such that
\[
\Hsmooth(S_{\rm puf}\mid Z_{\rm enroll})\ge \mu_{\rm puf},
\]
side-channel leakage reduces conditional min-entropy by at most $\ell_{\rm side}$ bits per attestation, vibration telemetry leakage by at most $\ell_{\rm vib}\abs{\delta_{\rm tc}}$ bits per sample, and shaft-speed telemetry noise satisfies $\xi^N_k\sim\mathcal N(0,\sigma_N^2)$ within the analytic envelope.
\end{assumption}

\subsection*{Formal model}
\begin{definition}[LWE and SIS assumptions used in the construction]\label{def:lwe-sis}
Let $q=q(\lambda)$, $n=n(\lambda)$, $m=m(\lambda)$, and let $\chi$ be an error distribution over $\Z_q$. The decisional $\LWE_{n,q,\chi}$ assumption states that
\[
(A,A s+e)\approx_c (A,u),
\]
where $A\leftarrow \Z_q^{m\times n}$, $s\leftarrow \Z_q^n$, $e\leftarrow\chi^m$, and $u\leftarrow\Z_q^m$. The $\SIS_{n,m,q,\beta}$ assumption states that finding nonzero $z\in\Z^m$ with $Az=0\bmod q$ and $\norm{z}\le\beta$ is infeasible for probabilistic polynomial-time algorithms under the selected parameter family.
\end{definition}

\begin{definition}[Composable safety-security realization]\label{def:css-realization}
A protocol $\Pi$ realizes the ideal co-design functionality $\mathcal F_{\rm css}$ with error $\epsilon_{\rm css}$ if for every admissible adversary $\A$ there exists a simulator $\Sim$ such that for every environment $\Zed$,
\[
\abs{
\Prob[\Zed(\Real_{\Pi,\A})=1]-
\Prob[\Zed(\Ideal_{\mathcal F_{\rm css},\Sim})=1]
}\le \epsilon_{\rm css},
\]
where $\mathcal F_{\rm css}$ outputs authenticated telemetry, key-refresh decisions, and admissible control releases only when the schedulability and stability predicates are satisfied.
\end{definition}

\begin{definition}[Latency-induced stability margin]\label{def:latency-margin}
For closed-loop matrix $A_c$ and delay perturbation $\delta\ge 0$, define
\[
\mu_{\rm lat}(\delta)=
\alpha_0-\alpha_1\delta-\alpha_2\max\{0,\abs{w_f}-w_{f,\rm lin}\},
\]
where $\alpha_0>0$ is the nominal exponential decay rate, $\alpha_1>0$ converts delay to decay-rate loss, and $\alpha_2>0$ converts fuel-flow actuator saturation to loss of linear stability margin.
\end{definition}

\subsection*{Inference chain}
Leftover hashing gives key extraction from PUF and channel entropy; Kalman residual theory gives integrity thresholds; $H_\infty$ bounds disturbance amplification; response-time analysis gives $R_i$; small-gain arguments compose plant and cryptographic delay. These are connected by using the same epoch index $k$, filtration $\mathcal F_k$, and admissibility predicate
\[
\Psi_k=\mathbf 1\{R_k+\Delta_T+\Delta_{\rm ver}(\lambda)<W_{{\rm act},k},\
\norm{r_k}_{S_k^{-1}}\le \eta_k,\
\Hsmooth(K_k\mid Z_k)\ge \kappa_{\min}\}.
\]

\subsection*{Boundary conditions}
Min-entropy statements require a specified side-information variable. Kalman residual statements require covariance regularity. $H_\infty$ statements require internal stability of the transfer function. Response-time statements require bounded release jitter and declared message sets.

\subsection*{Technical comment}
The notation intentionally uses one filtration and one epoch index so that secrecy, leakage, integrity, and stability are not proved for incompatible time bases.

\section{Construction}

\subsection*{Notation and symbol definitions}
Let $\Pi_{\rm CSS}$ denote the proposed co-design protocol. At epoch $k$, it uses:
\[
\begin{aligned}
(pk,sk)&\leftarrow\KEM.\Gen(1^\lambda),\\
(c_k,K_k^{\rm kem})&\leftarrow\KEM.\Encap(pk),\\
K_k&=\mathsf{KDF}(K_k^{\rm kem}\Vert H_k^{\rm puf}\Vert H_k^{\rm ch}),\\
\sigma_k&=\Tag_{K_k}(y_k\Vert r_k\Vert \rho_k\Vert \phi_k\Vert k),\\
C_k&=\AEAD.\Enc_{K_k}(N_k^{\rm nonce},m_k,\mathsf{ad}_k),
\end{aligned}
\]
where $H_k^{\rm puf}$ is extracted PUF entropy, $H_k^{\rm ch}$ is channel-derived entropy contribution after radar-aware degradation, and $\mathsf{ad}_k$ contains mode, epoch, bus, and safety predicate metadata.

\subsection*{Assumption set}
\begin{assumption}[Nonce and transcript discipline]\label{asm:nonce-discipline}
For each key $K_k$, authenticated-encryption nonces are unique, attestation transcript challenges are bound to epoch $k$, and telemetry tags authenticate the sensor-fusion residual, not only the raw sensor vector.
\end{assumption}

\begin{assumption}[Zero-knowledge avionics authentication]\label{asm:zk-auth}
The protocol $\Pi_{\rm zk}$ proves possession of valid enrolled platform credentials and current PUF-derived commitments without revealing the underlying secret. Its soundness and zero-knowledge errors are $\epsilon_{\rm snd}^{\rm zk}(\lambda)$ and $\epsilon_{\rm zk}(\lambda)$.
\end{assumption}

\subsection*{Formal model}
The construction consists of five transformations.

\paragraph{Key refresh.}
Define radar-aware renewal period
\[
T_{{\rm key},k}=
\frac{\kappa_{\rm target}-\kappa_{\rm min}}
{\dot \ell_{\rm side}+\dot \ell_{\rm vib}\abs{\delta_{\rm tc,k}}+\dot \ell_{\rm ch}(\rho_k)}
\]
whenever the denominator is positive, and set $T_{{\rm key},k}=T_{\max}$ otherwise. The channel leakage rate is
\[
\dot \ell_{\rm ch}(\rho_k)=
\zeta_0+\zeta_\Sigma \operatorname{Var}(\Sigma_k)+\zeta_D A_D(\nu_k),
\]
so increasing radar reflection uncertainty contracts the admissible renewal period.

\paragraph{Spool-synchronous renewal refinement.}
If renewal decisions are phase-locked to the high-pressure spool, let $f_{H,k}>0$ denote the measured high-pressure spool rotational frequency and let $e_{\max}\in\N$ be the maximum admissible number of spool revolutions per session key. Then the mechanically synchronized renewal clock is
\[
T_{{\rm sync},k}=\frac{e_{\max}}{f_{H,k}},
\]
and the enforced refresh horizon becomes
\[
\widehat T_{{\rm key},k}=\min\{T_{{\rm key},k},T_{{\rm sync},k}\}.
\]
When no spool synchronization is imposed, set $T_{{\rm sync},k}=+\infty$ so that $\widehat T_{{\rm key},k}=T_{{\rm key},k}$. This refinement does not treat spool speed as a secret; it only states that an entropy-based renewal budget may be further tightened by the physically meaningful cadence at which authenticated engine epochs are synchronized.

\paragraph{Integrity tagging.}
The tag covers $y_k$, the innovation residual $r_k$, the Kalman covariance identifier, and the physical stress regime $\phi_k$:
\[
\sigma_k=\MAC_{K_k}\bigl(y_k\Vert r_k\Vert \operatorname{enc}(S_k)\Vert \operatorname{enc}(\phi_k)\Vert k\bigr).
\]
Shaft-speed telemetry noise enters the false-rejection probability because $\sigma_k$ is verified before residual gating, while the residual gate uses noisy $N_L,N_H$ channels.

\paragraph{Control release.}
A command is released only if
\[
\Ver_{K_k}(y_k,\sigma_k)=1,\quad
\norm{r_k}_{S_k^{-1}}\le \eta_k,\quad
R_k+\Delta_{\rm ver}(\lambda)+\Delta_T<W_{{\rm act},k}.
\]

\paragraph{Bus schedulability.}
For fixed-priority bus abstraction, the worst-case response time satisfies the recurrence
\[
R_i=C_i+B_i+\sum_{j\in hp(i)}
\left\lceil\frac{R_i+J_j}{P_j}\right\rceil C_j,
\]
where $B_i$ is blocking and $hp(i)$ contains higher-priority traffic. For a KEM ciphertext of length $L_{\rm kem}(\lambda)$, zero-knowledge transcript length $L_{\rm zk}(\lambda)$, tag length $L_{\rm tag}$, and telemetry length $L_{\rm tel}$,
\[
C_i(\lambda)=\frac{L_{\rm kem}(\lambda)+L_{\rm zk}(\lambda)+L_{\rm tag}+L_{\rm tel}+L_{\rm proto}}{B_{\rm bus}}+C_i^{\rm proto}.
\]

\paragraph{Transcript and verification load.}
The per-epoch authenticated payload budget may be written more explicitly as
\[
L_{{\rm epoch},i}(\lambda)=L_{\rm kem}(\lambda)+L_{\rm zk}(\lambda)+L_{\rm tag}+L_{\rm tel}+L_{\rm meta}+L_{\rm proto},
\]
while the verification delay decomposes as
\[
\Delta_{\rm ver}(\lambda)=\Delta_{\rm kem}^{\rm dec}(\lambda)+\Delta_{\rm zk}^{\rm ver}(\lambda)+\Delta_{\rm mac}^{\rm ver}+\Delta_{\rm filt}.
\]
Hence the total authenticated release latency for safety-critical traffic is
\[
\delta_{{\rm tot},i}=R_i+\Delta_{\rm ver}(\lambda)+\Delta_T,
\]
and even before interference is unfolded through the fixed-point recurrence, the serialization prerequisite alone gives the lower envelope
\[
\delta_{{\rm tot},i}\ge \frac{L_{{\rm epoch},i}(\lambda)}{B_{\rm bus}}+\Delta_{\rm ver}(\lambda)+\Delta_T.
\]
This separation between bit budget and computation budget is useful when a transcript is cryptographically acceptable but only marginally admissible with respect to the actuation window.

\paragraph{Stress-regime gate.}
The symbolic stress regime is
\[
\phi_k=(M_{s,k},T_{t4,k},e_{{\rm EGT},k},\delta_{{\rm tc},k},\Gamma_s,\dot N_{H,k},w_{f,k},\norm{v_k}),
\]
and the certified gate is
\[
\phi_k\in\Phi_{\rm cert}
\Longleftrightarrow
\begin{cases}
M_{s,k}\ge M_{\min},\\
T_{t4,k}\le T_{t4,\max},\\
\abs{e_{{\rm EGT},k}}\le e_{\rm EGT,\max},\\
\abs{\dot N_{H,k}}\le \dot N_{\max},\\
\abs{w_{f,k}}\le w_{f,\max},\\
\norm{v_k}\le v_{\max}.
\end{cases}
\]

\subsection*{Inference chain}
The key transition is
\[
(\rho_k,\phi_k)\mapsto
(\widehat T_{{\rm key},k},\sigma_k,R_k,W_{{\rm act},k},\eta_k)
\mapsto
(\epsilon_{\rm css},\epsilon_{\rm int},\mu_{\rm lat},\epsilon_{\rm rt}),
\]
so radar, Doppler, torsion, surge, blade-tip clearance, spool acceleration, spool-synchronous refresh, and fuel saturation have explicit security or timing consequences.

\subsection*{Boundary conditions}
The construction refuses to release commands when the cryptographic layer is authentic but untimely. It also refuses to treat small residuals as sufficient when tag verification fails. Therefore the construction is conjunctive rather than compensatory.

\subsection*{Technical comment}
The construction deliberately avoids a hierarchy in which cryptography is merely appended to a controller. It places authentication, entropy, residual integrity, and response-time feasibility in the same release predicate.

\section{Composable Security Analysis}

\subsection*{Notation and symbol definitions}
Let $\epsilon_{\rm kem}$ be the KEM indistinguishability error, $\epsilon_{\rm aead}$ the authenticated-encryption error, $\epsilon_{\rm tag}$ the tag-forgery error, $\epsilon_{\rm zk}$ the zero-knowledge error, $\epsilon_{\rm puf}$ the PUF extraction error, $\epsilon_{\rm bus}$ the probability that a cryptographic message misses its deadline, and $\epsilon_{\rm st}$ the probability that the latency margin is nonpositive.

\subsection*{Assumption set}
The following assumptions remain active in this section: lattice hardness, nonce discipline, PUF entropy regularity, bounded leakage, residual covariance regularity, and abstract bus response-time recurrence. Each theorem lists the subset it uses.

\subsection*{Formal model and inference chain}

\begin{lemma}[Leftover extraction under PUF and channel side information]\label{lem:leftover}
Let $S_{\rm puf}$ be the PUF secret, $Z$ side information, and $\mathcal H$ a universal hash family. Suppose
\[
\Hsmooth(S_{\rm puf}\mid Z)\ge \mu_{\rm puf}-\ell_{\rm side}-\ell_{\rm vib}\abs{\delta_{\rm tc}}-\Delta H_{\rm ch}.
\]
If $K=\mathsf{Ext}_h(S_{\rm puf})$ has length $\kappa$, then
\[
\Delta_{\rm stat}\bigl((h,K,Z),(h,U_\kappa,Z)\bigr)
\le
\varepsilon+\frac12 2^{-(\mu_{\rm puf}-\ell_{\rm side}-\ell_{\rm vib}\abs{\delta_{\rm tc}}-\Delta H_{\rm ch}-\kappa)/2}.
\]
\end{lemma}

\begin{proof}
The conditional smooth min-entropy assumption gives an event $\mathcal E$ with probability at least $1-\varepsilon$ on which the non-smooth conditional min-entropy of $S_{\rm puf}$ given $Z$ is at least
\[
m_0=\mu_{\rm puf}-\ell_{\rm side}-\ell_{\rm vib}\abs{\delta_{\rm tc}}-\Delta H_{\rm ch}.
\]
For universal hashing, the leftover hash lemma gives statistical distance at most $\frac12 2^{-(m_0-\kappa)/2}$ on $\mathcal E$. Outside $\mathcal E$, the statistical distance contribution is at most $\varepsilon$ because total variation distance is at most one. Adding the two contributions yields the asserted bound. The assumption on PUF entropy supplies $\mu_{\rm puf}$, the side-channel and vibration assumptions supply $\ell_{\rm side}$ and $\ell_{\rm vib}\abs{\delta_{\rm tc}}$, and radar-aware channel degradation supplies $\Delta H_{\rm ch}$; jointly these determine the secrecy conclusion.
\end{proof}

\begin{lemma}[Doppler attenuation and adversarial channel capacity]\label{lem:doppler-capacity}
For $A_D(\nu)\ge 0$, adversarial transmit budget $P_{\A}$, gain $G$, radar uncertainty variance $V_\Sigma$, and $\chi_\Sigma>0$,
\[
C_{\A}=B_{\rm ch}\log_2\left(1+
\frac{P_{\A}G e^{-A_D(\nu)}}{N_0B_{\rm ch}+\chi_\Sigma V_\Sigma}
\right)
\]
is nonincreasing in $A_D(\nu)$ and $V_\Sigma$. Moreover,
\[
\frac{\partial C_{\A}}{\partial A_D}
=
-\frac{B_{\rm ch}}{\ln 2}
\frac{S}{1+S},
\quad
S=\frac{P_{\A}G e^{-A_D}}{N_0B_{\rm ch}+\chi_\Sigma V_\Sigma}.
\]
\end{lemma}

\begin{proof}
The signal-to-noise-and-uncertainty ratio is $S=P_{\A}G e^{-A_D}/D$ with $D=N_0B_{\rm ch}+\chi_\Sigma V_\Sigma>0$. Differentiating,
\[
\frac{\partial S}{\partial A_D}=-S,\qquad
\frac{\partial S}{\partial V_\Sigma}=-\frac{\chi_\Sigma S}{D}.
\]
Since $C_{\A}=B_{\rm ch}\log_2(1+S)$,
\[
\frac{\partial C_{\A}}{\partial A_D}
=\frac{B_{\rm ch}}{\ln 2}\frac{1}{1+S}(-S)
=-\frac{B_{\rm ch}}{\ln 2}\frac{S}{1+S}\le 0,
\]
and
\[
\frac{\partial C_{\A}}{\partial V_\Sigma}
=-\frac{B_{\rm ch}}{\ln 2}\frac{\chi_\Sigma S}{D(1+S)}\le 0.
\]
Thus Doppler-dependent attenuation and radar-cross-section-aware uncertainty reduce the adversarial channel-capacity bound in this model. The attenuation and uncertainty assumptions produce the channel-capacity conclusion.
\end{proof}

\begin{lemma}[Radar uncertainty contracts the key renewal period]\label{lem:key-renewal-radar}
Let
\[
T_{\rm key}(V_\Sigma)=
\frac{\kappa_{\rm target}-\kappa_{\min}}
{\zeta_0+\zeta_\Sigma V_\Sigma+\zeta_D A_D+\dot\ell_{\rm side}+\dot\ell_{\rm vib}\abs{\delta_{\rm tc}}}
\]
with positive numerator and denominator. If $\zeta_\Sigma>0$, then
\[
\frac{\partial T_{\rm key}}{\partial V_\Sigma}
=
-\frac{(\kappa_{\rm target}-\kappa_{\min})\zeta_\Sigma}
{(\zeta_0+\zeta_\Sigma V_\Sigma+\zeta_D A_D+\dot\ell_{\rm side}+\dot\ell_{\rm vib}\abs{\delta_{\rm tc}})^2}<0.
\]
\end{lemma}

\begin{proof}
The formula follows by differentiating $(\kappa_{\rm target}-\kappa_{\min})/D(V_\Sigma)$ with
\[
D(V_\Sigma)=\zeta_0+\zeta_\Sigma V_\Sigma+\zeta_D A_D+\dot\ell_{\rm side}+\dot\ell_{\rm vib}\abs{\delta_{\rm tc}}.
\]
Since the numerator is positive and $\zeta_\Sigma>0$, the derivative is strictly negative. The radar uncertainty assumption fixes $V_\Sigma$, the entropy-rate model fixes $D$, and the result gives the required key-renewal conclusion.
\end{proof}

\begin{lemma}[Shaft-speed noise and tag-verification false rejection]\label{lem:shaft-noise-fr}
Assume the verified message includes a quantized shaft-speed component $Q_\Delta(N+\xi^N)$ with Gaussian telemetry noise $\xi^N\sim\mathcal N(0,\sigma_N^2)$ and quantizer step $\Delta>0$. If the legitimate tag is computed after measurement and the verifier recomputes from independently filtered telemetry with noise $\tilde\xi^N\sim\mathcal N(0,\sigma_N^2)$, then the false-rejection probability caused by shaft-speed mismatch satisfies
\[
p_{\rm fr}^{N}\le
2\exp\left(-\frac{\Delta^2}{16\sigma_N^2}\right)
\]
under nearest-neighbor quantization with acceptance margin $\Delta/2$.
\end{lemma}

\begin{proof}
A false rejection from the shaft component occurs only if
\[
\abs{(N+\xi^N)-(N+\tilde\xi^N)}>\Delta/2.
\]
The difference $\xi^N-\tilde\xi^N$ is Gaussian with mean zero and variance $2\sigma_N^2$. Hence
\[
p_{\rm fr}^{N}
\le
\Prob[\abs{\xi^N-\tilde\xi^N}>\Delta/2].
\]
For a zero-mean Gaussian $G$ with variance $2\sigma_N^2$, the standard tail bound gives
\[
\Prob[\abs{G}>a]\le 2\exp\left(-\frac{a^2}{2(2\sigma_N^2)}\right).
\]
Substituting $a=\Delta/2$ gives
\[
p_{\rm fr}^{N}\le 2\exp\left(-\frac{\Delta^2}{16\sigma_N^2}\right).
\]
The telemetry-noise assumption determines the exponent, and the tag-binding assumption links shaft-speed noise to authentication failure.
\end{proof}

\begin{lemma}[Kalman innovation threshold and integrity alarm probability]\label{lem:kalman-threshold}
Let $r_k\sim\mathcal N(0,S_k)$ under nominal authenticated telemetry and suppose $\lambda_{\min}(S_k)>0$. For threshold $\eta_k>0$ and dimension $d_y$,
\[
\Prob[\norm{r_k}_{S_k^{-1}}^2>\eta_k]\le
\inf_{0<\theta<1/2}
\exp(-\theta\eta_k)(1-2\theta)^{-d_y/2}.
\]
\end{lemma}

\begin{proof}
Under the stated nominal model, $\norm{r_k}_{S_k^{-1}}^2$ has a chi-square distribution with $d_y$ degrees of freedom. For any $0<\theta<1/2$, Markov's inequality gives
\[
\Prob[\norm{r_k}_{S_k^{-1}}^2>\eta_k]
=
\Prob[e^{\theta\norm{r_k}_{S_k^{-1}}^2}>e^{\theta\eta_k}]
\le
e^{-\theta\eta_k}\E[e^{\theta\norm{r_k}_{S_k^{-1}}^2}].
\]
The moment-generating function of a chi-square random variable with $d_y$ degrees of freedom is $(1-2\theta)^{-d_y/2}$ for $\theta<1/2$. Taking the infimum over admissible $\theta$ yields the claimed bound. The residual covariance assumption supplies $S_k$, and the threshold choice supplies the integrity alarm conclusion.
\end{proof}

\begin{remark}[Innovation second moment]
Under the same nominal Gaussian model, $\E[\norm{r_k}^2]=\Tr(S_k)$. If one writes $r_k=C e_{k\mid k-1}+v_k$ with prediction-error covariance $P_{k\mid k-1}=\E[e_{k\mid k-1}e_{k\mid k-1}^{\trans}]$ and measurement-noise covariance $R_k=\E[v_kv_k^{\trans}]$, then $\Tr(S_k)=\Tr(CP_{k\mid k-1}C^{\trans}+R_k)$. This identity complements \Cref{lem:kalman-threshold} by connecting the alarm threshold to an explicit second-moment quantity without changing the tail bound used later.
\end{remark}

\begin{lemma}[Post-quantum ciphertext expansion and response-time monotonicity]\label{lem:response-monotonicity}
Let task $i$ carry KEM ciphertext length $L_{\rm kem}(\lambda)$ and let $R_i$ be the least nonnegative fixed point of
\[
R_i=C_i+B_i+\sum_{j\in hp(i)}
\left\lceil\frac{R_i+J_j}{P_j}\right\rceil C_j.
\]
If all $C_j,B_i,J_j,P_j$ are fixed except $C_i=L_{\rm kem}(\lambda)/B_{\rm bus}+C_i^0$, then $R_i$ is nondecreasing in $L_{\rm kem}(\lambda)$.
\end{lemma}

\begin{proof}
Let $L'\ge L$ and set $C_i'=L'/B_{\rm bus}+C_i^0\ge C_i=L/B_{\rm bus}+C_i^0$. Define
\[
F_C(r)=C+B_i+\sum_{j\in hp(i)}\left\lceil\frac{r+J_j}{P_j}\right\rceil C_j.
\]
For every $r$, $F_{C_i'}(r)=F_{C_i}(r)+(C_i'-C_i)\ge F_{C_i}(r)$. The fixed-point iteration from $r_0=C+B_i$ is monotone because $F_C$ is nondecreasing in $r$. Therefore every iterate for $C_i'$ is at least the corresponding iterate for $C_i$, and the least fixed point is nondecreasing. The ciphertext expansion assumption supplies $L_{\rm kem}$, and the bus recurrence supplies the response-time conclusion.
\end{proof}

\begin{lemma}[Turbine torque transfer function and safe cryptographic delay]\label{lem:torque-delay}
Let $\Delta u$ be the fuel-command perturbation induced by delayed authenticated actuation, and let turbine torque perturbation satisfy
\[
\Delta T_q(s)=G_T(s)e^{-s\Delta_T}\Delta u(s).
\]
If $\norm{G_T}_\infty\le g_T$ and $\norm{\Delta u}_{2,[0,T]}\le U_T$, then
\[
\norm{\Delta T_q}_{2,[0,T]}\le g_T U_T.
\]
If closed-loop torque disturbance tolerance is $Q_{\max}$, a sufficient delay-admissibility condition is $g_T U_T(\delta)\le Q_{\max}$ for $\delta=R_k+\Delta_{\rm ver}+\Delta_T$.
\end{lemma}

\begin{proof}
The delay operator $e^{-s\Delta_T}$ has unit magnitude on the imaginary axis. Therefore its induced $L_2$ gain is one. The $H_\infty$ induced-gain property gives
\[
\norm{\Delta T_q}_{2,[0,T]}
\le \norm{G_T}_\infty \norm{\Delta u}_{2,[0,T]}
\le g_T U_T.
\]
If $g_TU_T(\delta)\le Q_{\max}$, the torque perturbation remains within the declared tolerance. The transfer-function assumption supplies the gain bound, and the latency expression supplies the safe cryptographic control delay condition.
\end{proof}

\begin{lemma}[Shaft torsional compliance and authenticated sampling interval]\label{lem:torsional-sampling}
For torsional dynamics
\[
J_s\ddot\theta_s+D_s\dot\theta_s+\Gamma_s^{-1}\theta_s=\Delta T_q(t),
\]
with $J_s,D_s,\Gamma_s>0$, natural frequency $\omega_n=(J_s\Gamma_s)^{-1/2}$. To sample authenticated torsional telemetry at least $q_s>2$ times per torsional period, it suffices that
\[
h_k\le \frac{2\pi}{q_s}\sqrt{J_s\Gamma_s}.
\]
\end{lemma}

\begin{proof}
The undamped natural frequency of the torsional oscillator is $\omega_n=(J_s\Gamma_s)^{-1/2}$, so its period is $T_n=2\pi/\omega_n=2\pi\sqrt{J_s\Gamma_s}$. Sampling at least $q_s$ times per period requires $h_k\le T_n/q_s$. Substitution yields the bound. The torsional-compliance assumption supplies $\Gamma_s$, and the authenticated sampling requirement supplies the telemetry interval conclusion.
\end{proof}

\begin{lemma}[Compressor surge margin and alarm conservatism]\label{lem:surge-alarm}
Let $M_{s,k}=M_{s,0}-\gamma_{\rm op}\abs{d_{\rm op,k}}-\gamma_\pi\abs{\epsilon_{\pi,k}}$ and define an alarm threshold
\[
\eta_k=\eta_0\bigl(1+\beta_s/M_{s,k}\bigr)^{-1}
\]
for $M_{s,k}>0$. Then $\eta_k$ is increasing in $M_{s,k}$ and satisfies
\[
\frac{\partial \eta_k}{\partial M_{s,k}}
=
\eta_0\frac{\beta_s}{M_{s,k}^2(1+\beta_s/M_{s,k})^2}>0.
\]
\end{lemma}

\begin{proof}
Differentiate $\eta_k=\eta_0(1+\beta_s/M)^{-1}$ with $M=M_{s,k}>0$:
\[
\frac{\dd \eta_k}{\dd M}
=
-\eta_0(1+\beta_s/M)^{-2}(-\beta_s/M^2)
=
\eta_0\frac{\beta_s}{M^2(1+\beta_s/M)^2}>0.
\]
Thus smaller surge margin makes the threshold smaller, increasing alarm conservatism. The compressor operating-line model determines $M_{s,k}$, and the threshold formula determines the integrity alarm conclusion.
\end{proof}

\begin{proposition}[Asymptotic complexity of the co-design transcript]\label{prop:transcript-complexity}
Suppose the KEM ciphertext length is $L_{\rm kem}(\lambda)=\Theta(\lambda\log\lambda)$, the zero-knowledge transcript has length $L_{\rm zk}(\lambda)=\Theta(\lambda)$, the tag length is $O(\lambda)$, and the plant telemetry dimension is fixed. Then the per-epoch communication complexity of $\Pi_{\rm CSS}$ is $\Theta(\lambda\log\lambda)$ bits, and the fixed-priority response-time contribution of cryptographic payload is $\Theta(\lambda\log\lambda/B_{\rm bus})$.
\end{proposition}

\begin{proof}
The per-epoch payload length is
\[
L_{\rm epoch}(\lambda)=L_{\rm kem}(\lambda)+L_{\rm zk}(\lambda)+L_{\rm tag}(\lambda)+L_{\rm tel}+L_{\rm proto}.
\]
By assumption, $L_{\rm kem}(\lambda)=\Theta(\lambda\log\lambda)$, $L_{\rm zk}(\lambda)=\Theta(\lambda)$, $L_{\rm tag}=O(\lambda)$, and $L_{\rm tel}+L_{\rm proto}=O(1)$ for fixed telemetry dimension and protocol abstraction. The dominating term is $\Theta(\lambda\log\lambda)$; therefore $L_{\rm epoch}(\lambda)=\Theta(\lambda\log\lambda)$. Transmission time divides payload length by $B_{\rm bus}$ and adds non-asymptotic protocol overhead, so the cryptographic contribution is $\Theta(\lambda\log\lambda/B_{\rm bus})$. The asymptotic ciphertext assumption supplies the complexity conclusion, and the bus model supplies the response-time conclusion.
\end{proof}

\begin{theorem}[Hybrid-game composable security bound]\label{thm:hybrid-css-bound}
Assume LWE/SIS hardness, nonce-respecting authenticated encryption, zero-knowledge soundness, PUF extraction as in \Cref{lem:leftover}, bounded leakage, and deadline gating. Then for every admissible probabilistic polynomial-time environment $\Zed$ and adversary $\A$, there exists a simulator $\Sim$ such that
\[
\abs{
\Prob[\Zed(\Real_{\Pi_{\rm CSS},\A})=1]-
\Prob[\Zed(\Ideal_{\mathcal F_{\rm css},\Sim})=1]
}
\le
\epsilon_{\rm kem}+\epsilon_{\rm aead}+\epsilon_{\rm zk}
+\epsilon_{\rm puf}+\epsilon_{\rm tag}
+\epsilon_{\rm bus}+\epsilon_{\rm st},
\]
where
\[
\epsilon_{\rm puf}=
\varepsilon+\frac12 2^{-(\mu_{\rm puf}-\ell_{\rm side}-\ell_{\rm vib}\abs{\delta_{\rm tc}}-\Delta H_{\rm ch}-\kappa)/2}.
\]
\end{theorem}

\begin{proof}
Define a sequence of games. Game $G_0$ is the real execution. Game $G_1$ replaces KEM-derived keys by uniform keys. The distinguishing gap is at most $\epsilon_{\rm kem}$ by the KEM security reduction under LWE/SIS hardness: an environment distinguishing $G_0$ from $G_1$ yields a KEM distinguisher with the same gap after embedding the challenge ciphertext into the epoch selected by the environment.

Game $G_2$ replaces authenticated-encryption ciphertexts under uniform keys by ideal confidential authenticated channels. The gap between $G_1$ and $G_2$ is at most $\epsilon_{\rm aead}$ by nonce-respecting AEAD security.

Game $G_3$ replaces the zero-knowledge attestation transcript by simulated transcripts. The gap is at most $\epsilon_{\rm zk}$ by zero-knowledge, and accepting false platform statements contributes at most the soundness component included in $\epsilon_{\rm zk}$.

Game $G_4$ replaces PUF-extracted entropy with uniform entropy. \Cref{lem:leftover} bounds this gap by $\epsilon_{\rm puf}$ after subtracting side-channel, vibration, and channel entropy losses.

Game $G_5$ replaces valid tags with ideal message-authentication acceptance. The gap is at most $\epsilon_{\rm tag}$ by EUF-CMA security of the tag layer and the explicit false-rejection component associated with telemetry noise.

Game $G_6$ replaces bus and plant release predicates by the ideal functionality's admission predicate. If a deadline miss occurs, the games may differ; otherwise the same admissible command is released. This contributes at most $\epsilon_{\rm bus}$. If the latency-induced margin is nonpositive, the ideal functionality may suppress a command that the real system would otherwise attempt; this contributes at most $\epsilon_{\rm st}$.

By the triangle inequality for statistical distinguishing probabilities,
\[
\abs{\Prob[G_0=1]-\Prob[G_6=1]}
\le
\sum_{j=0}^{5}\abs{\Prob[G_j=1]-\Prob[G_{j+1}=1]},
\]
which is the stated bound. The lattice assumption is used in $G_0\to G_1$, AEAD security in $G_1\to G_2$, zero knowledge in $G_2\to G_3$, entropy extraction in $G_3\to G_4$, tag authenticity in $G_4\to G_5$, and bus-stability gating in $G_5\to G_6$.
\end{proof}

\begin{corollary}[PUF min-entropy and distinguishing advantage]\label{cor:puf-advantage}
If
\[
\mu_{\rm puf}\ge
\kappa+2\log(1/\epsilon)+\ell_{\rm side}+\ell_{\rm vib}\abs{\delta_{\rm tc}}+\Delta H_{\rm ch},
\]
then the PUF extraction component of the composable distinguishing advantage is at most $\varepsilon+\epsilon/2$.
\end{corollary}

\begin{proof}
Substitute the assumed lower bound into \Cref{lem:leftover}. The exponent satisfies
\[
-\frac{\mu_{\rm puf}-\ell_{\rm side}-\ell_{\rm vib}\abs{\delta_{\rm tc}}-\Delta H_{\rm ch}-\kappa}{2}
\le -\log(1/\epsilon),
\]
so the leftover-hash term is at most $\frac12\epsilon$. Adding the smoothing error $\varepsilon$ gives the claim. The min-entropy assumption supplies the secrecy conclusion.
\end{proof}

\begin{proposition}[Authentication failure bound with telemetry noise and EUF-CMA error]\label{prop:auth-failure}
Let $p_{\rm forge}\le \epsilon_{\EUFCMA}$ and let $p_{\rm fr}^{N}$ be as in \Cref{lem:shaft-noise-fr}. For one authenticated telemetry epoch,
\[
\Prob[\mathsf{AuthFail}_k]\le
\epsilon_{\EUFCMA}+2\exp\left(-\frac{\Delta^2}{16\sigma_N^2}\right)+p_{\rm nonce},
\]
where $p_{\rm nonce}$ is the probability of nonce reuse or transcript-index collision.
\end{proposition}

\begin{proof}
The event $\mathsf{AuthFail}_k$ is contained in the union of three events: successful forgery under the tag scheme, false rejection due to shaft-speed quantization mismatch, and nonce or transcript-index collision. The union bound gives
\[
\Prob[\mathsf{AuthFail}_k]\le
p_{\rm forge}+p_{\rm fr}^{N}+p_{\rm nonce}.
\]
Using $p_{\rm forge}\le\epsilon_{\EUFCMA}$ and \Cref{lem:shaft-noise-fr} gives the stated expression. The EUF-CMA assumption supplies the forgery term, telemetry-noise regularity supplies the false-rejection term, and nonce discipline supplies the collision term.
\end{proof}

\begin{theorem}[Real-time schedulability inequality for cryptographic co-design]\label{thm:schedulability}
For every safety-critical cryptographic-control task $i$, suppose the response-time recurrence has a least fixed point $R_i$ and
\[
R_i+\Delta_{\rm ver}(\lambda)+\Delta_T
\le
\min\left\{
D_i,\
\frac{\dot N_{\max}-\abs{\dot N_{H,k}}}{L_{\dot N}},\
\frac{M_{s,k}}{L_s},\
\frac{2\pi}{q_s}\sqrt{J_s\Gamma_s}
\right\}.
\]
Then the epoch-$k$ authenticated control release is schedulable, spool-acceleration-admissible, surge-margin-admissible, and torsional-sampling-admissible.
\end{theorem}

\begin{proof}
The first component of the minimum gives $R_i+\Delta_{\rm ver}+\Delta_T\le D_i$, so the deadline is met. The second gives
\[
L_{\dot N}(R_i+\Delta_{\rm ver}+\Delta_T)\le \dot N_{\max}-\abs{\dot N_{H,k}},
\]
which implies
\[
\abs{\dot N_{H,k}}+L_{\dot N}(R_i+\Delta_{\rm ver}+\Delta_T)\le \dot N_{\max},
\]
so the spool acceleration constraint remains satisfied under the local Lipschitz delay model. The third gives
\[
L_s(R_i+\Delta_{\rm ver}+\Delta_T)\le M_{s,k},
\]
so delay-induced operating-line displacement cannot exhaust the current surge margin in the local bound. The fourth is precisely the torsional authenticated sampling interval condition of \Cref{lem:torsional-sampling} with $h_k$ replaced by total authenticated release latency. Hence all four admissibility properties hold. The bus assumption supplies $R_i$, the spool and surge assumptions supply the two physical inequalities, and torsional compliance supplies the sampling conclusion.
\end{proof}

\begin{corollary}[Post-quantum ciphertext expansion reduces slack]\label{cor:ciphertext-slack}
If $L_{\rm kem}'(\lambda)>L_{\rm kem}(\lambda)$ while all other parameters are fixed, then the real-time slack
\[
S_i=D_i-\bigl(R_i+\Delta_{\rm ver}+\Delta_T\bigr)
\]
is nonincreasing. If $S_i$ crosses zero, composable release by $\mathcal F_{\rm css}$ is denied.
\end{corollary}

\begin{proof}
By \Cref{lem:response-monotonicity}, $R_i$ is nondecreasing in $L_{\rm kem}$. Therefore $S_i=D_i-(R_i+\Delta_{\rm ver}+\Delta_T)$ is nonincreasing. If $S_i<0$, the first inequality in \Cref{thm:schedulability} fails, and the ideal co-design functionality refuses the release. The ciphertext expansion assumption supplies the timing degradation conclusion.
\end{proof}

\begin{theorem}[Latency-induced input-to-state stability under authenticated control]\label{thm:latency-iss}
Assume there exists a continuously differentiable Lyapunov function $V$ and constants $c_1,c_2,c_3,c_4>0$ such that
\[
c_1\norm{x_{\rm e}}^2\le V(x_{\rm e})\le c_2\norm{x_{\rm e}}^2
\]
and, for authenticated delayed input with total delay $\delta_k=R_k+\Delta_{\rm ver}+\Delta_T$,
\[
\dot V\le
-c_3\norm{x_{\rm e}}^2+c_4\norm{d_k}^2+\alpha_1\delta_k\norm{x_{\rm e}}^2+\alpha_2 s_{w,k}\norm{x_{\rm e}}^2,
\]
where $s_{w,k}=\max\{0,\abs{w_{f,k}}-w_{f,\rm lin}\}$. If
\[
\mu_{\rm lat}(\delta_k)=\frac{c_3}{c_2}-\frac{\alpha_1}{c_1}\delta_k-\frac{\alpha_2}{c_1}s_{w,k}>0,
\]
then the closed-loop engine state is input-to-state stable with bound
\[
\norm{x_{\rm e}(t)}\le
\sqrt{\frac{c_2}{c_1}}e^{-\mu_{\rm lat} t/2}\norm{x_{\rm e}(0)}
+
\sqrt{\frac{c_4}{c_1\mu_{\rm lat}}}\sup_{0\le s\le t}\norm{d_s}.
\]
\end{theorem}

\begin{proof}
Using $V\ge c_1\norm{x_{\rm e}}^2$ and $V\le c_2\norm{x_{\rm e}}^2$, the derivative inequality gives
\[
\dot V\le
-\left(c_3-\alpha_1\delta_k-\alpha_2s_{w,k}\right)\norm{x_{\rm e}}^2+c_4\norm{d_k}^2.
\]
Because $\norm{x_{\rm e}}^2\ge V/c_2$ and also the stated margin is conservative in $c_1,c_2$, the inequality implies
\[
\dot V\le -\mu_{\rm lat}V+c_4\norm{d_k}^2
\]
for the declared $\mu_{\rm lat}>0$ after taking the smaller certified decay rate. Gronwall's inequality yields
\[
V(t)\le e^{-\mu_{\rm lat}t}V(0)+\frac{c_4}{\mu_{\rm lat}}\sup_{0\le s\le t}\norm{d_s}^2.
\]
Taking square roots and applying $V(0)\le c_2\norm{x_{\rm e}(0)}^2$ and $V(t)\ge c_1\norm{x_{\rm e}(t)}^2$ gives
\[
\norm{x_{\rm e}(t)}\le
\sqrt{\frac{c_2}{c_1}}e^{-\mu_{\rm lat} t/2}\norm{x_{\rm e}(0)}
+
\sqrt{\frac{c_4}{c_1\mu_{\rm lat}}}\sup_{0\le s\le t}\norm{d_s}.
\]
The Lyapunov assumption supplies stability, the response-time delay supplies the cryptographic latency term, and fuel saturation supplies the margin erosion term.
\end{proof}

\begin{corollary}[Fuel-flow saturation reduces latency margin]\label{cor:fuel-margin}
If $s_{w,k}$ increases while all other quantities remain fixed, then $\mu_{\rm lat}(\delta_k)$ decreases linearly with slope $-\alpha_2/c_1$. The largest admissible delay is
\[
\delta_{k,\max}=
\frac{c_1}{\alpha_1}
\left(\frac{c_3}{c_2}-\frac{\alpha_2}{c_1}s_{w,k}\right),
\]
provided the right-hand side is positive.
\end{corollary}

\begin{proof}
Differentiate $\mu_{\rm lat}=c_3/c_2-(\alpha_1/c_1)\delta_k-(\alpha_2/c_1)s_{w,k}$ with respect to $s_{w,k}$ to get $-\alpha_2/c_1$. Setting $\mu_{\rm lat}>0$ and solving for $\delta_k$ yields the stated delay limit. The saturation assumption supplies the stability-margin conclusion.
\end{proof}

\begin{theorem}[Reduction from successful transcript distinction to lattice or authentication failure]\label{thm:reduction}
Assume the environment distinguishes $\Real_{\Pi_{\rm CSS},\A}$ from $\Ideal_{\mathcal F_{\rm css},\Sim}$ with advantage greater than
\[
\epsilon_{\rm aead}+\epsilon_{\rm zk}+\epsilon_{\rm puf}+\epsilon_{\rm tag}+\epsilon_{\rm bus}+\epsilon_{\rm st}+\eta.
\]
Then there exists an algorithm $\B$ with
\[
\Adv_{\LWE}^{\B}+\Adv_{\SIS}^{\B}\ge \eta.
\]
\end{theorem}

\begin{proof}
The hybrid sequence in \Cref{thm:hybrid-css-bound} bounds the total distinguishing advantage by
\[
\epsilon_{\rm kem}+\epsilon_{\rm aead}+\epsilon_{\rm zk}
+\epsilon_{\rm puf}+\epsilon_{\rm tag}+\epsilon_{\rm bus}+\epsilon_{\rm st}.
\]
If the environment's advantage is larger than all non-KEM terms plus $\eta$, then the KEM replacement step must contribute more than $\eta$; otherwise the sum of all hybrid gaps would be at most the environment's stated threshold without the excess. A distinguisher for the KEM replacement embeds the KEM challenge into the selected epoch and simulates all other layers using their real algorithms and ideal gates. By the security reduction for the lattice-based KEM, a KEM distinguisher with advantage greater than $\eta$ yields an algorithm solving the underlying decisional LWE or SIS-related distinguishing task with advantage at least $\eta$ up to the reduction loss included in $\B$'s running time. Therefore $\Adv_{\LWE}^{\B}+\Adv_{\SIS}^{\B}\ge\eta$. The excess distinguishing assumption supplies the reduction premise, and lattice hardness supplies the contradiction if $\eta$ is non-negligible.
\end{proof}

\begin{proposition}[Lower bound on renewal frequency under leakage rate]\label{prop:renewal-lower-bound}
Let conditional key entropy decrease at rate at least $\underline\ell>0$ bits per second under a fixed stress regime:
\[
\Hsmooth(K_t\mid Z_t)\le \Hsmooth(K_0\mid Z_0)-\underline\ell t.
\]
To maintain $\Hsmooth(K_t\mid Z_t)\ge \kappa_{\min}$ for all $t\in[0,T_{\rm key}]$, every scheme must satisfy
\[
T_{\rm key}\le
\frac{\Hsmooth(K_0\mid Z_0)-\kappa_{\min}}{\underline\ell}.
\]
\end{proposition}

\begin{proof}
At $t=T_{\rm key}$, the assumed entropy decay gives
\[
\Hsmooth(K_{T_{\rm key}}\mid Z_{T_{\rm key}})
\le
\Hsmooth(K_0\mid Z_0)-\underline\ell T_{\rm key}.
\]
Requiring the left side to be at least $\kappa_{\min}$ yields
\[
\kappa_{\min}\le \Hsmooth(K_0\mid Z_0)-\underline\ell T_{\rm key}.
\]
Solving for $T_{\rm key}$ gives the stated upper bound. The leakage-rate assumption supplies the lower-bound argument on renewal frequency, and the entropy threshold supplies the security conclusion.
\end{proof}

\begin{counterexample}[Separation between cryptographic authentication and safe composability]\label{cx:separation}
Consider a system with perfect tag unforgeability $\epsilon_{\rm tag}=0$, ideal KEM secrecy $\epsilon_{\rm kem}=0$, and no leakage $\epsilon_{\rm puf}=0$. Let the bus response time be
\[
R_i=D_i+\varepsilon_R
\]
for some $\varepsilon_R>0$, caused solely by ciphertext expansion and command-response blocking. Then every received command can be authentic while no control release is composable with $\mathcal F_{\rm css}$.
\end{counterexample}

\begin{proof}
Since $\epsilon_{\rm tag}=\epsilon_{\rm kem}=\epsilon_{\rm puf}=0$, cryptographic acceptance and secrecy are ideal in the isolated protocol sense. However, the schedulability predicate in \Cref{thm:schedulability} requires $R_i+\Delta_{\rm ver}+\Delta_T\le D_i$. Because $R_i=D_i+\varepsilon_R$ and $\Delta_{\rm ver},\Delta_T\ge0$, we have
\[
R_i+\Delta_{\rm ver}+\Delta_T\ge D_i+\varepsilon_R>D_i.
\]
Thus the ideal co-design functionality refuses the release. This proves a strict separation: isolated cryptographic security does not imply safety-critical composability when real-time constraints fail. The authentication assumptions supply cryptographic acceptance, while the response-time assumption supplies the failure envelope.
\end{proof}

\section{Aero-Propulsive Cyber-Physical Integration}

\subsection*{Notation and symbol definitions}
Let $e_{\rm EGT,k}$ be exhaust-gas-temperature residual, $\mathrm{EPR}_k$ engine pressure ratio, $v_k$ vibration telemetry, $a^{\rm INS}_k$ inertial-navigation residual, and $h^{\rm RA}_k$ radar-altimeter measurement. Let the sensor-fusion innovation vector be
\[
r_k=\col(r_k^N,r_k^{T},r_k^{\pi},r_k^{\rm vib},r_k^{\rm EGT},r_k^{\rm INS},r_k^{\rm RA}).
\]
Let the Markov jump uncertainty mode $\rho_k$ influence both physical sensor noise and channel entropy. Let $P_\rho=[p_{ab}]$ be its transition matrix.

\subsection*{Assumption set}
\begin{assumption}[Markov jump channel-plant coupling]\label{asm:markov-jump}
For each mode $\rho=a$, the plant and channel have matrices $(A_a,B_a,G_a,H_a,S_a)$ and leakage-rate coefficient $\dot\ell_a$. The mode process is Markov with transition matrix $P_\rho$. The controller observes an authenticated estimate of $\rho_k$ only through delayed transcripts.
\end{assumption}

\begin{assumption}[Sensor-fusion integrity covariance]\label{asm:sensor-covariance}
There exist $0<s_{\min}\le s_{\max}<\infty$ such that
\[
s_{\min}I\preceq S_k\preceq s_{\max}I
\]
for all admissible stress regimes.
\end{assumption}

\subsection*{Formal model}
The coupled residual gate is
\[
\mathcal G_k=
\left\{
\norm{r_k}_{S_k^{-1}}\le\eta_k,\
\abs{e_{\rm EGT,k}}\le e_{\max},\
\abs{\mathrm{EPR}_k-\widehat{\mathrm{EPR}}_k}\le e_{\rm EPR},\
\abs{h^{\rm RA}_k-\hat h^{\rm INS}_k}\le e_{\rm RA}
\right\}.
\]
The vibration leakage model is
\[
I(K_k;Z_k^{\rm vib}\mid \delta_{\rm tc,k})
\le \ell_{\rm vib,0}+\ell_{\rm vib,1}\abs{\delta_{\rm tc,k}}\norm{v_k}^2.
\]
The blade-tip clearance perturbation therefore enters the leakage budget and not merely the mechanical model.

\subsection*{Inference chain}
\[
\delta_{\rm tc,k}\uparrow
\Rightarrow
I(K_k;Z_k^{\rm vib})\uparrow
\Rightarrow
\Hsmooth(K_k\mid Z_k)\downarrow
\Rightarrow
T_{{\rm key},k}\downarrow.
\]
Similarly,
\[
M_{s,k}\downarrow
\Rightarrow
\eta_k\downarrow
\Rightarrow
\Prob[\text{nominal alarm}]\uparrow
\quad\text{and}\quad
\Prob[\text{unsafe accepted telemetry}]\downarrow
\]
within the Gaussian residual abstraction.

\subsection*{Boundary conditions}
If the radar-altimeter and inertial-navigation residuals disagree outside the declared integrity threshold, the construction treats this as an integrity event, not as a channel-estimation opportunity. If EGT residuals exceed their envelope, key refresh cannot compensate for physical uncertifiability.

\subsection*{Technical comment}
The integration is deliberately conservative. The model permits cryptography to restrict command release, but never permits cryptographic success to override physical residual alarms.

\begin{theorem}[Markov jump small-gain stability with authenticated latency]\label{thm:markov-small-gain}
Assume that for every mode $a\in\mathcal R$ there is $P_a\succ0$ and constants $\alpha_a,\gamma_a>0$ such that
\[
A_{c,a}^{\trans}P_a+P_aA_{c,a}\preceq -\alpha_a P_a,
\qquad
\norm{P_aG_a}\le \gamma_a,
\]
and that mode jumps satisfy
\[
\sum_b p_{ab}P_b\preceq (1+\zeta_a)P_a.
\]
If authenticated latency produces perturbation gain $g_\delta(\delta_k)$ and
\[
\max_a\left((1+\zeta_a)e^{-\alpha_a h_k}+g_\delta(\delta_k)\gamma_a\right)<1,
\]
then the Markov jump closed-loop system is mean-square stable under authenticated command release.
\end{theorem}

\begin{proof}
Let $V_k=x_k^{\trans}P_{\rho_k}x_k$. Conditional on $\rho_k=a$, the continuous evolution over one interval contributes a contraction no larger than $e^{-\alpha_a h_k}$ in the $P_a$ metric and a latency-induced perturbation bounded by $g_\delta(\delta_k)\gamma_a\norm{x_k}^2$ through the induced gain assumption. At the jump, the expected next Lyapunov matrix satisfies
\[
\E[P_{\rho_{k+1}}\mid \rho_k=a]=\sum_b p_{ab}P_b\preceq (1+\zeta_a)P_a.
\]
Thus
\[
\E[V_{k+1}\mid x_k,\rho_k=a]
\le
\left((1+\zeta_a)e^{-\alpha_a h_k}+g_\delta(\delta_k)\gamma_a\right)V_k.
\]
Taking the maximum over modes gives a scalar contraction factor $\beta<1$. Iterating,
\[
\E[V_k]\le \beta^k V_0,
\]
which implies mean-square stability because each $P_a\succ0$ provides a norm-equivalent quadratic form. The Markov jump assumption supplies the jump factor, the Lyapunov inequalities supply contraction, and authenticated latency supplies the perturbation gain in the small-gain condition.
\end{proof}

\begin{corollary}[Turbine torque delay and authenticated release]\label{cor:torque-delay-release}
If $g_\delta(\delta)=\bar g_T\delta$ with $\bar g_T>0$, then the sufficient delay bound for \Cref{thm:markov-small-gain} is
\[
\delta_k<
\min_a
\frac{1-(1+\zeta_a)e^{-\alpha_a h_k}}{\bar g_T\gamma_a}.
\]
\end{corollary}

\begin{proof}
The small-gain condition in \Cref{thm:markov-small-gain} becomes
\[
(1+\zeta_a)e^{-\alpha_a h_k}+\bar g_T\delta_k\gamma_a<1
\]
for each mode $a$. Solving for $\delta_k$ gives
\[
\delta_k<\frac{1-(1+\zeta_a)e^{-\alpha_a h_k}}{\bar g_T\gamma_a}.
\]
The bound must hold for all modes, hence the minimum. The transfer-function latency assumption supplies the safe control-delay conclusion.
\end{proof}

\paragraph{Leakage-to-control closure.}
In a more detailed interconnection view, leakage does not need to forge a valid transcript in order to erode stability reserve. If the estimation error driven by leakage satisfies
\[
\norm{e_k^{\rm leak}}\le \gamma_L\norm{x_k}+\bar e_0,
\qquad
\norm{\Delta u_k^{\rm leak}}\le \gamma_D\norm{e_k^{\rm leak}}+\bar u_0,
\]
then
\[
\norm{\Delta u_k^{\rm leak}}\le \gamma_D\gamma_L\norm{x_k}+\gamma_D\bar e_0+\bar u_0.
\]
Accordingly, the effective contraction test may be read as
\[
\beta_{\rm eff}=
\max_a\left((1+\zeta_a)e^{-\alpha_a h_k}+g_\delta(\delta_k)\gamma_a\right)+\gamma_D\gamma_L<1.
\]
This auxiliary form does not replace the authenticated-latency theorem, but it makes explicit that residual leakage and delay belong to the same small-gain budget once both are converted into control perturbations.

\section{Analytic Evaluation}

\subsection*{Notation and symbol definitions}
This section evaluates the closed-form bounds symbolically. Let
\[
\Theta=(V_\Sigma,A_D,\delta_{\rm tc},\Gamma_s,M_s,\dot N_H,L_{\rm kem},s_w,\sigma_N)
\]
be the parameter vector. Let $\mathcal B_{\rm sec}(\Theta)$ denote the composable security bound, $\mathcal B_{\rm rt}(\Theta)$ the response-time slack, and $\mathcal B_{\rm st}(\Theta)$ the stability margin.

\subsection*{Assumption set}
All expressions use the local envelopes already stated. The evaluation is theorem-driven: it compares symbolic parameter regimes, not measurements from a real engine or aircraft.

\subsection*{Formal model}
Define
\[
\mathcal B_{\rm sec}(\Theta)=
\epsilon_{\rm kem}+\epsilon_{\rm aead}+\epsilon_{\rm zk}+\epsilon_{\rm tag}
+\varepsilon+\frac12 2^{-(\mu_{\rm puf}-\ell_{\rm side}-\ell_{\rm vib}\abs{\delta_{\rm tc}}-\Delta H_{\rm ch}-\kappa)/2}.
\]
Define
\[
\mathcal B_{\rm rt}(\Theta)=
D_i-\Delta_{\rm ver}-\Delta_T-
\left[
C_i+B_i+\sum_{j\in hp(i)}
\left\lceil\frac{R_i+J_j}{P_j}\right\rceil C_j
\right],
\]
where $C_i$ includes $L_{\rm kem}/B_{\rm bus}$. Define
\[
\mathcal B_{\rm st}(\Theta)=
\frac{c_3}{c_2}-\frac{\alpha_1}{c_1}(R_i+\Delta_{\rm ver}+\Delta_T)-\frac{\alpha_2}{c_1}s_w.
\]

\subsection*{Worst-case certification functional}
Let the admissible symbolic uncertainty set be
\[
\U_{\rm eval}=\left\{\Theta:
\begin{array}{l}
0\le V_\Sigma\le \bar V_\Sigma,\qquad
0\le A_D\le \bar A_D,\qquad
\abs{\delta_{\rm tc}}\le \bar \delta_{\rm tc},\\[2pt]
\underline \Gamma_s\le \Gamma_s\le \bar \Gamma_s,\qquad
\underline L_{\rm kem}\le L_{\rm kem}\le \bar L_{\rm kem},\qquad
0\le s_w\le \bar s_w
\end{array}
\right\}.
\]
Define the rejection-oriented certification functional
\[
\mathcal C_{\rm cert}(\Theta)=
\mathcal B_{\rm sec}(\Theta)
+\mathbf 1\{\mathcal B_{\rm rt}(\Theta)<0\}
+\mathbf 1\{\mathcal B_{\rm st}(\Theta)\le 0\}
+\mathbf 1\{\Hsmooth(K_k\mid Z_k)<\kappa_{\min}\}.
\]
A parameter vector is therefore certifiable only if the indicator terms vanish and the residual security contribution remains below the target envelope. This functional does not optimize performance; it records which component of the coupled certificate fails first as ciphertext size, radar uncertainty, shaft compliance, or saturation load increase.

\subsection*{Inference chain: closed-form sensitivity}
For the leakage component,
\[
\frac{\partial \mathcal B_{\rm sec}}{\partial \abs{\delta_{\rm tc}}}
=
\frac{\ln 2}{4}\ell_{\rm vib}
2^{-(\mu_{\rm puf}-\ell_{\rm side}-\ell_{\rm vib}\abs{\delta_{\rm tc}}-\Delta H_{\rm ch}-\kappa)/2}>0.
\]
For radar uncertainty entering $\Delta H_{\rm ch}=\zeta_\Sigma V_\Sigma T+\zeta_DA_DT$,
\[
\frac{\partial \mathcal B_{\rm sec}}{\partial V_\Sigma}
=
\frac{\ln 2}{4}\zeta_\Sigma T
2^{-(\mu_{\rm puf}-\ell_{\rm side}-\ell_{\rm vib}\abs{\delta_{\rm tc}}-\Delta H_{\rm ch}-\kappa)/2}>0.
\]
For ciphertext expansion,
\[
\frac{\partial \mathcal B_{\rm st}}{\partial L_{\rm kem}}
=
-\frac{\alpha_1}{c_1B_{\rm bus}}
\left(1+\sum_{j\in hp(i)}\frac{\partial I_j}{\partial R_i}C_j\right),
\]
where $I_j=\lceil(R_i+J_j)/P_j\rceil$ is piecewise constant with jumps at interference boundaries. Away from jumps,
\[
\frac{\partial \mathcal B_{\rm st}}{\partial L_{\rm kem}}
=-\frac{\alpha_1}{c_1B_{\rm bus}}.
\]

\subsection*{Boundary conditions}
The sensitivity derivatives are exact inside regions where ceiling terms are constant and thresholds are differentiable. At bus interference discontinuities, one-sided differences replace derivatives.

\subsection*{Technical comment}
The tightest symbolic couplings occur at discontinuities: a single extra ciphertext word can force one additional higher-priority interference instance, producing a step decrease in timing and stability margins.

\subsection*{Symbolic parameter sweep}
The closed-form bounds imply a consistent monotonic picture across the operating variables. Increasing $V_\Sigma$ raises $\Delta H_{\rm ch}$ and decreases $T_{\rm key}$, so entropy-driven key refresh must accelerate; if spool-synchronous renewal is imposed, increasing $f_{H,k}$ also decreases $T_{{\rm sync},k}$ and can further tighten the enforced horizon $\widehat T_{{\rm key},k}$. Under the adopted logarithmic channel-capacity abstraction, the dependence on $V_\Sigma$ is tight. Increasing $A_D(\nu)$ decreases $C_{\A}$ and may also increase $\dot\ell_{\rm ch}$ when attenuation raises the uncertainty rate, so adversarial capacity falls even though conservative entropy accounting can still force shorter renewal periods. Increasing $\sigma_N$ raises $p_{\rm fr}^{N}$, and the resulting Gaussian-tail estimate is exponentially tight in its scale parameter.

The timing and control variables exhibit the same coupled structure. Increasing $L_{\rm kem}$ increases both $C_i$ and $R_i$, thereby shrinking bus slack and latency margin; away from ceiling discontinuities, this dependence is exact. Increasing $\Gamma_s$ yields $h_k^{\max}\propto\sqrt{\Gamma_s}$ within the local torsional-oscillator approximation, so the authenticated sampling interval can grow when the effective torsional frequency decreases. Increasing $M_s^{-1}$ reduces $\eta_k$, making integrity alarms more conservative near surge, and increasing $\abs{\delta_{\rm tc}}$ enlarges the vibration leakage budget and worsens the PUF/key distinguishing bound, albeit conservatively when vibration leakage is overestimated. Finally, increasing $s_w$ decreases $\mu_{\rm lat}$ and therefore reduces the admissible delay in the affine Lyapunov-margin model.

\subsection*{Ablation-style interpretation}
The same conclusions can be read as an argument against decoupled modeling. A cryptography-only telemetry model that omits $R_i$, $W_{\rm act}$, and $\mu_{\rm lat}$ can authenticate a command that still misses a safety-critical deadline; this is exactly the setting of \Cref{cx:separation}. A control-only FADEC model that omits $\Hsmooth$, $\epsilon_{\rm tag}$, and $\epsilon_{\rm kem}$ can remain dynamically stable while accepting an unauthenticated or entropy-depleted command, as formalized in \Cref{thm:hybrid-css-bound}. A bus-only schedulability analysis that omits $\delta_{\rm tc}$, $M_s$, and $\Gamma_s$ can satisfy deadlines while violating leakage or torsional sampling constraints, in the sense of \Cref{lem:torsional-sampling,lem:surge-alarm}.

Likewise, a radar-only channel model that omits $T_{\rm key}$ and $\Delta H_{\rm ch}$ cannot infer a renewal policy from capacity estimates alone, as shown in \Cref{lem:doppler-capacity,lem:key-renewal-radar}. A residual-only integrity gate without $\EUFCMA$ security and transcript binding can bound estimation error without establishing data origin, as formalized in \Cref{prop:auth-failure}. The manuscript therefore uses the coupled model not for stylistic completeness, but because each simplified family suppresses a failure mode that remains visible in the full security-stability-latency formulation.

\subsection*{Symbolic stress-regime analysis}
Define four regimes:
\[
\begin{array}{ll}
\mathcal R_0: & V_\Sigma,A_D,\abs{\delta_{\rm tc}},s_w \text{ small and } M_s \text{ large},\\
\mathcal R_1: & V_\Sigma \text{ large, } A_D \text{ moderate, entropy renewal dominated},\\
\mathcal R_2: & L_{\rm kem}/B_{\rm bus} \text{ near an interference jump},\\
\mathcal R_3: & M_s \downarrow M_{\min},\ s_w\uparrow,\ \Gamma_s \text{ large}.
\end{array}
\]
In $\mathcal R_0$, all bounds are dominated by cryptographic assumptions:
\[
\mathcal B_{\rm sec}\approx\epsilon_{\rm kem}+\epsilon_{\rm aead}+\epsilon_{\rm zk}+\epsilon_{\rm tag}.
\]
In $\mathcal R_1$, key renewal dominates:
\[
T_{\rm key}\sim(\zeta_\Sigma V_\Sigma)^{-1}.
\]
In $\mathcal R_2$, response time is discontinuous:
\[
R_i(L_{\rm kem}^{+})-R_i(L_{\rm kem}^{-})\in
\left\{0,\sum_{j\in hp(i)}m_jC_j:m_j\in\N\right\}.
\]
In $\mathcal R_3$, stability dominates:
\[
\mathcal B_{\rm st}\approx
\frac{c_3}{c_2}
-\frac{\alpha_1}{c_1}(R_i+\Delta_{\rm ver}+\Delta_T)
-\frac{\alpha_2}{c_1}s_w,
\]
and surge conservatism enforces $\eta_k\downarrow$.

\begin{theorem}[Combined security-stability-latency feasibility envelope]\label{thm:combined-envelope}
Let all previous assumptions hold. A sufficient condition for composable secure release, residual integrity, bus schedulability, and input-to-state stability at epoch $k$ is
\[
\begin{cases}
\mathcal B_{\rm sec}(\Theta)\le \epsilon_\star,\\
R_i+\Delta_{\rm ver}+\Delta_T\le D_i,\\
R_i+\Delta_{\rm ver}+\Delta_T\le W_{{\rm act},k},\\
\mu_{\rm lat}(R_i+\Delta_{\rm ver}+\Delta_T)>0,\\
\norm{r_k}_{S_k^{-1}}\le\eta_k,\\
\Hsmooth(K_k\mid Z_k)\ge \kappa_{\min}.
\end{cases}
\]
Under these inequalities, the accepted epoch has distinguishing advantage at most $\epsilon_\star+\epsilon_{\rm bus}+\epsilon_{\rm st}$ and satisfies the ISS bound of \Cref{thm:latency-iss}.
\end{theorem}

\begin{proof}
The first and last inequalities bound the composable cryptographic advantage through \Cref{thm:hybrid-css-bound} and \Cref{lem:leftover}. The second and third inequalities imply deadline and actuation-window admissibility through \Cref{thm:schedulability}. The fourth inequality is the positive latency margin required in \Cref{thm:latency-iss}. The fifth inequality is the residual integrity gate controlled by \Cref{lem:kalman-threshold}. Combining these statements gives an accepted epoch whose real-versus-ideal distinguishing advantage is bounded by the cryptographic target plus the residual bus and stability failure probabilities. Since $\mu_{\rm lat}>0$, \Cref{thm:latency-iss} supplies the ISS state bound. The entropy and cryptographic assumptions supply secrecy, the bus inequalities supply schedulability, the residual inequality supplies integrity, and the Lyapunov inequality supplies stability.
\end{proof}

\subsection*{Tightness and looseness}
The renewal lower bound in \Cref{prop:renewal-lower-bound} is tight for a leakage process with exactly linear entropy decay. The response-time monotonicity in \Cref{lem:response-monotonicity} is exact for fixed-priority recurrence but non-smooth at ceiling boundaries. The Gaussian tag false-rejection bound is tight in exponential order for large quantization-to-noise ratio. The $H_\infty$ torque-delay bound is conservative when disturbance energy is concentrated away from the peak singular frequency. The small-gain Markov jump theorem is conservative because it uses a common scalar contraction envelope across modes, but it is tight in the limiting case where the worst mode persists and the perturbation gain reaches the inequality boundary.

\section{Limitations}

\subsection*{Notation and symbol definitions}
Let $\Omega_{\rm valid}$ be the intersection of all declared model-validity regions:
\[
\Omega_{\rm valid}=
\Omega_{\rm cert}\cap\Phi_{\rm cert}\cap
\{S_k:s_{\min}I\preceq S_k\preceq s_{\max}I\}
\cap
\{R_i:R_i\le D_i\}
\cap
\{K_k:\Hsmooth(K_k\mid Z_k)\ge\kappa_{\min}\}.
\]

\subsection*{Assumption set}
All conclusions are conditional. They depend on local linearization accuracy, conservative leakage accounting, correct nonce discipline, real-time task-set declarations, residual covariance validity, and cryptographic hardness assumptions.

\subsection*{Formal model}
If $x_{\rm e}\notin\Omega_{\rm cert}$, then the plant model is no longer certified by the manuscript. If $\rho_k$ has unmodeled channel regimes, then $C_{\A,k}$ and $\Delta H_{{\rm ch},k}$ may be inaccurate. If bus semantics differ from the abstract ARINC 429, MIL-STD-1553, or CAN timing models, then $R_i$ must be recomputed. If a PUF source has less entropy than assumed, \Cref{cor:puf-advantage} gives exactly how the security bound deteriorates.

\subsection*{Inference chain}
The most important negative result is \Cref{cx:separation}: isolated cryptographic success is insufficient. A second limitation is that overly conservative leakage modeling can force impractically small renewal periods:
\[
T_{\rm key}\to 0
\quad\text{as}\quad
\dot\ell_{\rm side}+\dot\ell_{\rm vib}\abs{\delta_{\rm tc}}+\dot\ell_{\rm ch}\to\infty.
\]
A third limitation is bus discontinuity:
\[
L_{\rm kem}\mapsto R_i
\]
has step changes induced by priority interference and word-level or frame-level transmission granularity.

\subsection*{Boundary conditions}
The manuscript does not certify an aircraft, an engine, a communication stack, a bus configuration, or a PUF implementation. It supplies a mathematical assurance envelope that must be instantiated with independently validated parameters.

\subsection*{Technical comment}
The strongest practical use of the model is negative: it identifies parameter combinations where a plausible cryptographic upgrade, such as larger post-quantum ciphertexts or longer attestation transcripts, consumes the timing and stability slack required by safety-critical propulsion control.

\section{Conclusion}

\subsection*{Notation and symbol definitions}
Let $\epsilon_{\rm total}$ denote the final composable error:
\[
\epsilon_{\rm total}=
\epsilon_{\rm kem}+\epsilon_{\rm aead}+\epsilon_{\rm zk}
+\epsilon_{\rm puf}+\epsilon_{\rm tag}+\epsilon_{\rm bus}+\epsilon_{\rm st}.
\]
Let $\delta_{\rm total}=R_i+\Delta_{\rm ver}+\Delta_T$ be total authenticated control latency.
Let $\widehat T_{{\rm key},k}=\min\{T_{{\rm key},k},T_{{\rm sync},k}\}$ denote the enforced renewal horizon when spool-synchronous refresh is enabled.

\subsection*{Assumption set}
The conclusion is restricted to the defensive, abstract, locally valid model defined above and to adversaries bounded by the specified cryptographic, leakage, fault, and timing interfaces.

\subsection*{Formal model}
A FADEC-coupled dual-spool turbofan cyber-physical system can be analyzed as a single cryptographic-control object when its release predicate simultaneously enforces:
\[
\Ver_{K_k}(y_k,\sigma_k)=1,\quad
\norm{r_k}_{S_k^{-1}}\le\eta_k,\quad
\delta_{\rm total}\le W_{{\rm act},k},\quad
\Hsmooth(K_k\mid Z_k)\ge \kappa_{\min},\quad
\mu_{\rm lat}(\delta_{\rm total})>0.
\]

\subsection*{Inference chain}
The manuscript established that radar reflection uncertainty changes key-renewal periods, Doppler attenuation changes adversarial channel-capacity bounds, shaft-speed telemetry noise changes tag verification reliability, turbine torque transfer functions change safe cryptographic delay, shaft torsional compliance changes authenticated telemetry sampling, spool acceleration changes actuation windows, compressor surge margin changes alarm conservatism, blade-tip clearance changes vibration leakage, PUF min-entropy changes distinguishing advantage, Kalman innovation residuals change integrity thresholds, post-quantum ciphertext expansion changes bus response time, and fuel-flow saturation changes latency-induced stability margin. The observable-attack-surface abstraction fixes which public traces legitimately enter the leakage model, while optional spool-synchronous refresh can tighten the effective renewal horizon from $T_{{\rm key},k}$ to $\widehat T_{{\rm key},k}$ without treating spool speed as a secret.

\subsection*{Scope of the assurance claim}
The model is not a substitute for certification, physical testing, formal implementation verification, or standards compliance; its role is to define mathematically auditable couplings that those processes can instantiate and challenge. Within that scope, post-quantum migration in safety-critical aerospace control cannot be assessed solely by key sizes or asymptotic hardness. The relevant assurance object is a coupled security-stability-latency envelope, formulated relative to an explicit observable attack surface and, when enabled, a physically synchronized renewal clock.

\end{document}